\documentclass{article}

\usepackage[preprint]{nips_2018}

\usepackage[utf8]{inputenc} 
\usepackage[T1]{fontenc}    
\usepackage{hyperref}       
\usepackage{url}            
\usepackage{booktabs}       
\usepackage{algorithm}
\usepackage{algpseudocode}
\usepackage{amsfonts}       
\usepackage{amsmath}
\usepackage{amsthm}
\usepackage{nicefrac}       
\usepackage{microtype}      
\usepackage{graphicx}

\title{Fast Greedy MAP Inference for Determinantal Point Process to Improve Recommendation Diversity}

\author{
  Laming Chen \\
  Hulu LLC \\
  Beijing, China \\
  \texttt{laming.chen@hulu.com} \\
  \And
  Guoxin Zhang\thanks{This work was conducted while the author was with Hulu.} \\
  Kuaishou \\
  Beijing, China \\
  \texttt{zgx.net@qq.com} \\
  \And
  Hanning Zhou \\
  Hulu LLC \\
  Beijing, China \\
  \texttt{eric.zhou@hulu.com} \\
}

\DeclareMathOperator{\argmax}{arg\,max}
\DeclareMathOperator*{\mean}{mean}

\newtheorem{theorem}{Theorem}

\begin{document}

\maketitle

\begin{abstract}
The determinantal point process (DPP) is an elegant probabilistic model of repulsion
with applications in various machine learning tasks including summarization and search.
However, the maximum a posteriori (MAP) inference for DPP which plays an important role in many applications is NP-hard,
and even the popular greedy algorithm can still be too computationally expensive to be used in large-scale real-time scenarios.
To overcome the computational challenge, in this paper, we propose a novel algorithm to greatly accelerate the greedy MAP inference for DPP.
In addition, our algorithm also adapts to scenarios where the repulsion is only required among nearby few items in the result sequence.
We apply the proposed algorithm to generate relevant and diverse recommendations.
Experimental results show that our proposed algorithm is significantly faster than state-of-the-art competitors,
and provides a better relevance-diversity trade-off on several public datasets, which is also confirmed in an online A/B test.
\end{abstract}

\section{Introduction}
\label{sec:intro}

The determinantal point process (DPP) was first introduced in \cite{macchi1975coincidence} to give the distributions of fermion systems in thermal equilibrium.
The repulsion of fermions is described precisely by DPP, making it natural for modeling diversity.
Besides its early applications in quantum physics and random matrices \cite{mehta1960density},
it has also been recently applied to various machine learning tasks such as multiple-person pose estimation \cite{kulesza2010structured},
image search \cite{kulesza2011k}, document summarization \cite{kulesza2011learning},
video summarization \cite{gong2014diverse}, product recommendation \cite{gillenwater2014expectation}, and tweet timeline generation \cite{yao2016tweet}.
Compared with other probabilistic models such as the graphical models, one primary advantage of DPP is that it admits polynomial-time algorithms for many types of inference,
including conditioning and sampling \cite{kulesza2012determinantal}.

One exception is the important maximum a posteriori (MAP) inference, i.e., finding the set of items with the highest probability, which is NP-hard \cite{ko1995exact}.
Consequently, approximate inference methods with low computational complexity are preferred.
A near-optimal MAP inference method for DPP is proposed in \cite{gillenwater2012near}.
However, this algorithm is a gradient-based method with high computational complexity for evaluating the gradient in each iteration,
making it impractical for large-scale real-time applications.
Another method is the widely used greedy algorithm \cite{nemhauser1978analysis}, justified by the fact that the log-probability of set in DPP is submodular.
Despite its relatively weak theoretical guarantees \cite{ccivril2009selecting},
it is widely used due to its promising empirical performance \cite{kulesza2011learning,gong2014diverse,yao2016tweet}.
Known exact implementations of the greedy algorithm \cite{gillenwater2012near,li2016gaussian} have $O(M^4)$ complexity, where $M$ is the total number of items.
Han et al.'s recent work \cite{han2017faster} reduces the complexity down to $O(M^3)$ by introducing some approximations, which sacrifices accuracy.
In this paper, we propose an \textit{exact} implementation of the greedy algorithm with $O(M^3)$ complexity,
and it runs much faster than the approximate one \cite{han2017faster} empirically.

The essential characteristic of DPP is that it assigns higher probability to sets of items that are diverse from each other \cite{kulesza2012determinantal}.
In some applications, the selected items are displayed as a sequence, and the negative interactions are restricted only among nearby few items.
For example, when recommending a long sequence of items to the user, each time only a small portion of the sequence catches the user's attention.
In this scenario, requiring items far away from each other to be diverse is unnecessary.
Developing fast algorithm for this scenario is another motivation of this paper.

\textbf{Contributions.} In this paper, we propose a novel algorithm to greatly accelerate the greedy MAP inference for DPP.
By updating the Cholesky factor incrementally, our algorithm reduces the complexity down to $O(M^3)$,
and runs in $O(N^2M)$ time to return $N$ items, making it practical to be used in large-scale real-time scenarios.
To the best of our knowledge, this is the first exact implementation of the greedy MAP inference for DPP with such a low time complexity.

In addition, we also adapt our algorithm to scenarios where the diversity is only required within a sliding window.
Supposing the window size is $w<N$, the complexity can be reduced to $O(wNM)$.
This feature makes it particularly suitable for scenarios where we need a long sequence of items diversified within a short sliding window.

Finally, we apply our proposed algorithm to the recommendation task.
Recommending diverse items gives the users exploration opportunities to discover novel and serendipitous items, and also enables the service to discover users' new interests.
As shown in the experimental results on public datasets and an online A/B test,
the DPP-based approach enjoys a favorable trade-off between relevance and diversity compared with the known methods.

\section{Background and Related Work}

\textbf{Notations.} Sets are represented by uppercase letters such as $Z$, and $\#Z$ denotes the number of elements in $Z$.
Vectors and matrices are represented by bold lowercase letters and bold uppercase letters, respectively.
$(\cdot)^{\top}$ denotes the transpose of the argument vector or matrix.
$\langle {\bf x},{\bf y}\rangle$ is the inner product of two vectors $\bf x$ and $\bf y$.
Given subsets $X$ and $Y$, ${\bf L}_{X,Y}$ is the sub-matrix of $\bf L$ indexed by $X$ in rows and $Y$ in columns.
For notation simplicity, we let ${\bf L}_{X,X}={\bf L}_{X}$, ${\bf L}_{X,\{i\}}={\bf L}_{X,i}$, and ${\bf L}_{\{i\},X}={\bf L}_{i,X}$.
$\det({\bf L})$ is the determinant of $\bf L$, and $\det({\bf L}_\emptyset)=1$ by convention.

\subsection{Determinantal Point Process}

DPP is an elegant probabilistic model with the ability to express negative interactions \cite{kulesza2012determinantal}.
Formally, a DPP $\mathcal{P}$ on a discrete set $Z=\{1, 2, \ldots, M\}$ is a probability measure on $2^Z$, the set of all subsets of $Z$.
When $\mathcal{P}$ gives nonzero probability to the empty set,
there exists a matrix ${\bf L}\in\mathbb{R}^{M\times M}$ such that for every subset $Y\subseteq Z$, the probability of $Y$ is $\mathcal{P}(Y)\propto\det({\bf L}_{Y})$,
where ${\bf L}$ is a real, positive semidefinite (PSD) kernel matrix indexed by the elements of $Z$.
Under this distribution, many types of inference tasks including marginalization, conditioning, and sampling can be performed in polynomial time,
except for the MAP inference
\begin{displaymath}
  Y_{\mathrm{map}}=\argmax_{Y\subseteq Z}\det({\bf L}_{Y}).
\end{displaymath}
In some applications, we need to impose a cardinality constraint on $Y$ to return a subset of fixed size with the highest probability,
resulting in the MAP inference for $k$-DPP \cite{kulesza2011k}.

Besides the works on the MAP inference for DPP introduced in Section~\ref{sec:intro},
some other works propose to draw samples and return the one with the highest probability.
In \cite{gillenwater2014approximate}, a fast sampling algorithm with complexity $O(N^2M)$ is proposed when the eigendecomposition of $\bf L$ is available.
Though \cite{gillenwater2014approximate} and our work both aim to accelerate existing algorithms,
the methodology is essentially different: we rely on incrementally updating the Cholesky factor.

\subsection{Diversity of Recommendation}

Improving the recommendation diversity has been an active field in machine learning and information retrieval.
Some works addressed this problem in a generic setting to achieve better trade-off between arbitrary relevance and dissimilarity functions
\cite{carbonell1998use,bradley2001improving,zhang2008avoiding,borodin2012max,he2012gender}.
However, they used only pairwise dissimilarities to characterize the overall diversity property of the list,
which may not capture some complex relationships among items
(e.g., the characteristics of one item can be described as a simple linear combination of another two).
Some other works tried to build new recommender systems to promote diversity through the learning process \cite{ahmed2012fair,su2013set,xia2017adapting},
but this makes the algorithms less generic and unsuitable for direct integration into existing recommender systems.

Some works proposed to define the similarity metric based on the taxonomy information
\cite{ziegler2005improving,agrawal2009diversifying,chandar2013preference,vargas2014coverage,ashkan2015optimal,teo2016adaptive}.
However, the semantic taxonomy information is not always available, and it may be unreliable to define similarity based on them.
Several other works proposed to define the diversity metric based on explanation \cite{yu2009takes},
clustering \cite{boim2011diversification,aytekin2014clustering,lee2017single}, feature space \cite{qin2013promoting},
or coverage \cite{wu2016relevance,puthiya2016coverage}.

In this paper, we apply the DPP model and our proposed algorithm to optimize the trade-off between relevance and diversity.
Unlike existing techniques based on pairwise dissimilarities, our method defines the diversity in the feature space of the entire subset.
Notice that our approach is essentially different from existing DPP-based methods for recommendation.
In \cite{gillenwater2014expectation,mariet2015fixed,gartrell2016bayesian,gartrell2017low},
they proposed to recommend complementary products to the ones in the shopping basket,
and the key is to learn the kernel matrix of DPP to characterize the relations among items.
By contrast, we aim to generate a relevant and diverse recommendation list through the MAP inference.

The diversity considered in our paper is different from the aggregate diversity in \cite{adomavicius2011maximizing,niemann2013new}.
Increasing aggregate diversity promotes long tail items, while improving diversity prefers diverse items in each recommendation list.

\section{Fast Greedy MAP Inference}
\label{sec:fast}

In this section, we present a fast implementation of the greedy MAP inference algorithm for DPP.
In each iteration, item
\begin{equation}\label{dpp:map}
  j=\argmax_{i\in{Z\setminus Y_{\mathrm{g}}}}\log\det({\bf L}_{Y_{\mathrm{g}}\cup\{i\}})-\log\det({\bf L}_{Y_{\mathrm{g}}})
\end{equation}
is added to $Y_{\mathrm{g}}$.
Since $\bf L$ is a PSD matrix, all of its principal minors are also PSD.
Suppose $\det({\bf L}_{Y_{\mathrm{g}}})>0$, and the Cholesky decomposition of ${\bf L}_{Y_{\mathrm{g}}}$ is ${\bf L}_{Y_{\mathrm{g}}}={\bf V}{\bf V}^\top$,
where ${\bf V}$ is an invertible lower triangular matrix.
For any $i\in{Z\setminus Y_{\mathrm{g}}}$, the Cholesky decomposition of ${\bf L}_{Y_{\mathrm{g}}\cup\{i\}}$ can be derived as
\begin{equation}\label{cho:dec:i}
  {\bf L}_{Y_{\mathrm{g}}\cup\{i\}}=
  \begin{bmatrix}
    {\bf L}_{Y_{\mathrm{g}}} & {\bf L}_{Y_{\mathrm{g}},i} \\
    {\bf L}_{i,Y_{\mathrm{g}}} & {\bf L}_{ii}
  \end{bmatrix}=
  \begin{bmatrix}
    {\bf V} & {\bf 0} \\
    {\bf c}_i & d_i
  \end{bmatrix}
  \begin{bmatrix}
    {\bf V} & {\bf 0} \\
    {\bf c}_i & d_i
  \end{bmatrix}^{\top},
\end{equation}
where row vector ${\bf c}_i$ and scalar $d_i\ge0$ satisfies
\begin{align}
  {\bf V}{\bf c}_i^{\top}&={\bf L}_{Y_{\mathrm{g}},i},\label{ci:direct} \\
  d_i^2&={\bf L}_{ii}-\|{\bf c}_i\|_2^2.\label{di:direct}
\end{align}
In addition, according to Equ. \eqref{cho:dec:i}, it can be derived that
\begin{equation}\label{di:intep}
  \det({\bf L}_{Y_{\mathrm{g}}\cup\{i\}})=\det({\bf V}{\bf V}^{\top})\cdot d_i^2=\det({\bf L}_{Y_{\mathrm{g}}})\cdot d_i^2.
\end{equation}
Therefore, Opt. \eqref{dpp:map} is equivalent to
\begin{equation}\label{dpp:fast}
  j=\argmax_{i\in{Z\setminus Y_{\mathrm{g}}}}\log(d_i^2).
\end{equation}
Once Opt. \eqref{dpp:fast} is solved, according to Equ. \eqref{cho:dec:i}, the Cholesky decomposition of ${\bf L}_{Y_{\mathrm{g}}\cup\{j\}}$ becomes
\begin{equation}\label{chole:decom}
  {\bf L}_{Y_{\mathrm{g}}\cup\{j\}}=
  \begin{bmatrix}
    {\bf V} & {\bf 0} \\
    {\bf c}_j & d_j
  \end{bmatrix}
  \begin{bmatrix}
    {\bf V} & {\bf 0} \\
    {\bf c}_j & d_j
  \end{bmatrix}^{\top},
\end{equation}
where ${\bf c}_j$ and $d_j$ are readily available.
The Cholesky factor of ${\bf L}_{Y_{\mathrm{g}}}$ can therefore be efficiently updated after a new item is added to $Y_{\mathrm{g}}$.

For each item $i$, ${\bf c}_i$ and $d_i$ can also be updated incrementally.
After Opt. \eqref{dpp:fast} is solved, define ${\bf c}'_i$ and ${d}'_i$ as the new vector and scalar of $i\in Z\setminus(Y_{\mathrm{g}}\cup\{j\})$.
According to Equ.~\eqref{ci:direct} and Equ. \eqref{chole:decom}, we have
\begin{equation}\label{ci:equation}
  \begin{bmatrix}
    {\bf V} & {\bf 0} \\
    {\bf c}_j & d_j
  \end{bmatrix}
  {{\bf c}'_i}^{\top}
  ={\bf L}_{Y_{\mathrm{g}}\cup\{j\},i}
  =\begin{bmatrix}
    {\bf L}_{Y_{\mathrm{g}},i} \\
    {\bf L}_{ji}
  \end{bmatrix}.
\end{equation}
Combining Equ. \eqref{ci:equation} with Equ. \eqref{ci:direct}, we conclude
\begin{displaymath}
  {\bf c}'_i=\begin{bmatrix}
    {\bf c}_i & ({\bf L}_{ji}-\langle{\bf c}_j,{\bf c}_i\rangle)/d_j
  \end{bmatrix}
  \doteq \begin{bmatrix}{\bf c}_i & e_i\end{bmatrix}.
\end{displaymath}
Then Equ. \eqref{di:direct} implies
\begin{equation}\label{di:update}
  d_i^{\prime2}={\bf L}_{ii}-\|{\bf c}'_i\|_2^2={\bf L}_{ii}-\|{\bf c}_i\|_2^2-e_i^2=d_i^2-e_i^2.
\end{equation}

Initially, $Y_{\mathrm{g}}=\emptyset$, and Equ. \eqref{di:intep} implies $d_i^2=\det({\bf L}_{ii})={\bf L}_{ii}$.
The complete algorithm is summarized in Algorithm~\ref{sup:fast}.
The \textit{stopping criteria} is $d_j^2<1$ for unconstraint MAP inference or $\#Y_{\textrm{g}}>N$ when the cardinality constraint is imposed.
For the latter case, we introduce a small number $\varepsilon>0$ and add $d_j^2<\varepsilon$ to the \textit{stopping criteria} for numerical stability of calculating $1/d_j$.

\begin{algorithm}[tb]
	\caption{Fast Greedy MAP Inference}
	\label{sup:fast}
	\begin{algorithmic}[1]
		\State {\bfseries Input:} Kernel $\bf L$, \textit{stopping criteria}
		\State {\bfseries Initialize:} ${\bf c}_i=[]$, $d_i^2={\bf L}_{ii}$, $j=\argmax_{i\in{Z}}\log(d_i^2)$, $Y_{\textrm{g}}=\{j\}$
		\While{\textit{stopping criteria} not satisfied}
		\For{$i\in{Z\setminus Y_{\mathrm{g}}}$}
		\State $e_i=({\bf L}_{ji}-\langle{\bf c}_j,{\bf c}_i\rangle)/d_j$
		\State ${\bf c}_i=[{\bf c}_i \quad e_i]$, $d_i^2=d_i^2-e_i^2$
		\EndFor
		\State $j=\argmax_{i\in{Z\setminus Y_{\mathrm{g}}}}\log(d_i^2)$, $Y_{\textrm{g}}=Y_{\textrm{g}}\cup\{j\}$
		\EndWhile
		\State {\bfseries Return:} $Y_{\mathrm{g}}$
	\end{algorithmic}
\end{algorithm}

In the $k$-th iteration, for each item $i\in Z\setminus Y_{\mathrm{g}}$, updating ${\bf c}_i$ and $d_i$ involve the inner product of two vectors of length $k$,
resulting in overall complexity $O(kM)$.
Therefore, Algorithm~\ref{sup:fast} runs in $O(M^3)$ time for unconstraint MAP inference and $O(N^2M)$ to return $N$ items.
Notice that this is achieved by additional $O(NM)$ (or $O(M^2)$ for the unconstraint case) space for ${\bf c}_i$ and $d_i$.

\section{Diversity within Sliding Window}

In some applications, the selected set of items are displayed as a sequence, and the diversity is only required within a sliding window.
Denote the window size as $w$.
We modify Opt. \eqref{dpp:map} to
\begin{equation}\label{map:local}
  j=\argmax_{i\in{Z\setminus Y_{\mathrm{g}}}}\log\det({\bf L}_{Y_{\mathrm{g}}^w\cup\{i\}})-\log\det({\bf L}_{Y_{\mathrm{g}}^w}),
\end{equation}
where $Y_{\mathrm{g}}^w\subseteq Y_{\mathrm{g}}$ contains $w-1$ most recently added items.
When $\#Y_{\mathrm{g}}\ge w$, a simple modification of method \cite{li2016gaussian} solves Opt. \eqref{map:local} with complexity $O(w^2M)$.
We adapt our algorithm to this scenario so that Opt. \eqref{map:local} can be solved in $O(wM)$ time.

In Section~\ref{sec:fast}, we showed how to efficiently select a new item when $\bf V$, ${\bf c}_i$, and $d_i$ are available.
For Opt. \eqref{map:local}, ${\bf V}$ is the Cholesky factor of ${\bf L}_{Y_{\mathrm{g}}^w}$.
After Opt. \eqref{map:local} is solved, we can similarly update $\bf V$, ${\bf c}_i$, and $d_i$ for ${\bf L}_{Y_{\mathrm{g}}^w\cup\{j\}}$.
When the number of items in $Y_{\mathrm{g}}^w$ is $w-1$, to update $Y_{\mathrm{g}}^w$,
we also need to remove the earliest added item in $Y_{\mathrm{g}}^w$.
The detailed derivations of updating ${\bf V}$, ${\bf c}_i$, and $d_i$ when the earliest added item is removed are given in the supplementary material.

\begin{algorithm}[t]
	\caption{Fast Greedy MAP Inference with a Sliding Window}
	\label{diverse:nearby}
	\begin{algorithmic}[1]
		\State {\bfseries Input:} Kernel $\bf L$, window size $w$, \textit{stopping criteria}
		\State {\bfseries Initialize:} ${\bf V}=[]$, ${\bf c}_i=[]$, $d_i^2={\bf L}_{ii}$, $j=\argmax_{i\in{Z}}\log(d_i^2)$, $Y_{\textrm{g}}=\{j\}$
		\While{\textit{stopping criteria} not satisfied}
		\State Update ${\bf V}$ according to Equ. \eqref{chole:decom}
		\For{$i\in{Z\setminus Y_{\mathrm{g}}}$}
		\State $e_i=({\bf L}_{ji}-\langle{\bf c}_j,{\bf c}_i\rangle)/d_j$
		\State ${\bf c}_i=[{\bf c}_i \quad e_i]$, $d_i^2=d_i^2-e_i^2$
		\EndFor
		\If{$\#Y_{\mathrm{g}}\ge w$}
		\State ${\bf v}={\bf V}_{2:,1}$, ${\bf V}={\bf V}_{2:}$, $a_i={\bf c}_{i,1}$, ${\bf c}_i={\bf c}_{i,2:}$
		\For{$l=1,\cdots,w-1$}
		\State $t^2={\bf V}_{ll}^2+{\bf v}_l^2$
		\State ${\bf V}_{l+1:,l}=({\bf V}_{l+1:,l}{\bf V}_{ll}+{\bf v}_{l+1:}{\bf v}_l)/t$, ${\bf v}_{l+1:}=({\bf v}_{l+1:}t-{\bf V}_{l+1:,l}{\bf v}_l)/{\bf V}_{ll}$
		\For{$i\in{Z\setminus Y_{\mathrm{g}}}$}
		\State ${\bf c}_{i,l}=({\bf c}_{i,l}{\bf V}_{ll}+a_{i}{\bf v}_l)/t$, $a_i=(a_it-{\bf c}_{i,l}{\bf v}_l)/{\bf V}_{ll}$
		\EndFor
		\State ${\bf V}_{ll}=t$
		\EndFor
		\For{$i\in{Z\setminus Y_{\mathrm{g}}}$}
		\State $d_i^2=d_i^2+a_i^2$
		\EndFor
		\EndIf
		\State $j=\argmax_{i\in{Z\setminus Y_{\mathrm{g}}}}\log(d_i^2)$, $Y_{\textrm{g}}=Y_{\textrm{g}}\cup\{j\}$
		\EndWhile
		\State {\bfseries Return:} $Y_{\mathrm{g}}$
	\end{algorithmic}
\end{algorithm}

The complete algorithm is summarized in Algorithm~\ref{diverse:nearby}.
Line 10-21 shows how to update $\bf V$, ${\bf c}_i$, and $d_i$ in place after the earliest item is removed.
In the $k$-th iteration where $k\ge w$, updating $\bf V$, all ${\bf c}_i$ and $d_i$ require $O(w^2)$, $O(wM)$, and $O(M)$ time, respectively.
The overall complexity of Algorithm~\ref{diverse:nearby} is $O(wNM)$ to return $N\ge w$ items.
Numerical stability is discussed in the supplementary material.
\section{Improving Recommendation Diversity}

In this section, we describe a DPP-based approach for recommending relevant and diverse items to users.
For a user $u$, the profile item set $P_u$ is defined as the set of items that the user likes.
Based on $P_u$, a recommender system recommends items $R_u$ to the user.

The approach takes three inputs: a candidate item set $C_u$, a score vector ${\bf r}_u$ which indicates how relevant the items in $C_u$ are,
and a PSD matrix $\bf S$ which measures the similarity of each pair of items.
The first two inputs can be obtained from the internal results of many traditional recommendation algorithms.
The third input, similarity matrix $\bf S$, can be obtained based on the attributes of items, the interaction relations with users, or a combination of both.
This approach can be regarded as a ranking algorithm balancing the relevance of items and their similarities.

To apply the DPP model in the recommendation task, we need to construct the kernel matrix.
As revealed in \cite{kulesza2012determinantal}, the kernel matrix can be written as a Gram matrix,
${\bf L}={\bf B}^{\top}{\bf B}$, where the columns of $\bf B$ are vectors representing the items.
We can construct each column vector ${\bf B}_i$ as the product of the item score $r_i\ge0$
and a normalized vector ${\bf f}_i\in\mathbb{R}^D$ with $\|{\bf f}_i\|_2=1$.
The entries of kernel $\bf L$ can be written as
\begin{equation}\label{kernel:entry}
  {\bf L}_{ij}=\langle{\bf B}_i,{\bf B}_j\rangle=\langle r_i{\bf f}_i,r_j{\bf f}_j\rangle=r_ir_j\langle{\bf f}_i,{\bf f}_j\rangle.
\end{equation}
We can think of $\langle{\bf f}_i,{\bf f}_j\rangle$ as measuring the similarity between item $i$ and item $j$, i.e., $\langle{\bf f}_i,{\bf f}_j\rangle={\bf S}_{ij}$.
Therefore, the kernel matrix for user $u$ can be written as ${\bf L}=\mathrm{Diag}({\bf r}_u)\cdot{\bf S}\cdot\mathrm{Diag}({\bf r}_u)$,
where $\mathrm{Diag}({\bf r}_u)$ is a diagonal matrix whose diagonal vector is ${\bf r}_u$.
The log-probability of $R_u$ is
\begin{equation}\label{dpp:decom}
  \log\det({\bf L}_{R_u})=\sum_{i\in R_u}\log({\bf r}_{u,i}^2)+\log\det({\bf S}_{R_u}).
\end{equation}
The second term in Equ. \eqref{dpp:decom} is maximized when the item representations of $R_u$ are orthogonal, and therefore it promotes diversity.
It clearly shows how the DPP model incorporates the relevance and diversity of the recommended items.

A nice feature of methods in \cite{carbonell1998use,zhang2008avoiding,borodin2012max}
is that they involve a tunable parameter which allows users to adjust the trade-off between relevance and diversity.
According to Equ.~\eqref{dpp:decom}, the original DPP model does not offer such a mechanism.
We modify the log-probability of $R_u$ to
\begin{displaymath}
  \log\mathcal{P}(R_u)\propto\theta\cdot\sum_{i\in R_u}{\bf r}_{u,i}+(1-\theta)\cdot\log\det({\bf S}_{R_u}),
\end{displaymath}
where $\theta\in[0,1]$.
This corresponds to a DPP with kernel ${\bf L}'=\mathrm{Diag}(\exp(\alpha{\bf r}_u))\cdot{\bf S}\cdot\mathrm{Diag}(\exp(\alpha{\bf r}_u))$,
where $\alpha=\theta/(2(1-\theta))$.
We can also get the marginal gain of log-probability $\log\mathcal{P}(R_u\cup\{i\})-\log\mathcal{P}(R_u)$ as
\begin{equation}\label{marg:new}
  \theta\cdot{\bf r}_{u,i}+(1-\theta)\cdot(\log\det({\bf S}_{R_u\cup\{i\}})-\log\det({\bf S}_{R_u})).
\end{equation}
Then Algorithm~\ref{sup:fast} and Algorithm~\ref{diverse:nearby} can be easily modified to maximize \eqref{marg:new} with kernel matrix $\bf S$.

Notice that we need the similarity ${\bf S}_{ij}\in[0,1]$ for the recommendation task, where $0$ means the most diverse and $1$ means the most similar.
This may be violated when the inner product of normalized vectors $\langle{\bf f}_i,{\bf f}_j\rangle$ can take negative values.
In the extreme case, the most diverse pair ${\bf f}_i=-{\bf f}_j$, but the determinant of the corresponding sub-matrix is $0$, same as ${\bf f}_i={\bf f}_j$.
To guarantee nonnegativity, we can take a linear mapping while keeping $\bf S$ a PSD matrix, e.g.,
\begin{displaymath}
  {\bf S}_{ij}=\frac{1+\langle{\bf f}_i,{\bf f}_j\rangle}{2}=\left\langle\frac{1}{\sqrt{2}}\begin{bmatrix} 1 \\ {\bf f}_i \end{bmatrix},\frac{1}{\sqrt{2}}\begin{bmatrix} 1 \\ {\bf f}_j \end{bmatrix}\right\rangle\in[0,1].
\end{displaymath}

\section{Experimental Results}

In this section, we evaluate and compare our proposed algorithms on synthetic dataset and real-world recommendation tasks.
Algorithms are implemented in Python with vectorization.
The experiments are performed on a laptop with $2.2$GHz Intel Core i$7$ and $16$GB RAM.

\subsection{Synthetic Dataset}

\begin{figure}[t]
\begin{center}
	\begin{minipage}{.24\textwidth}
		\includegraphics[width=1.35in]{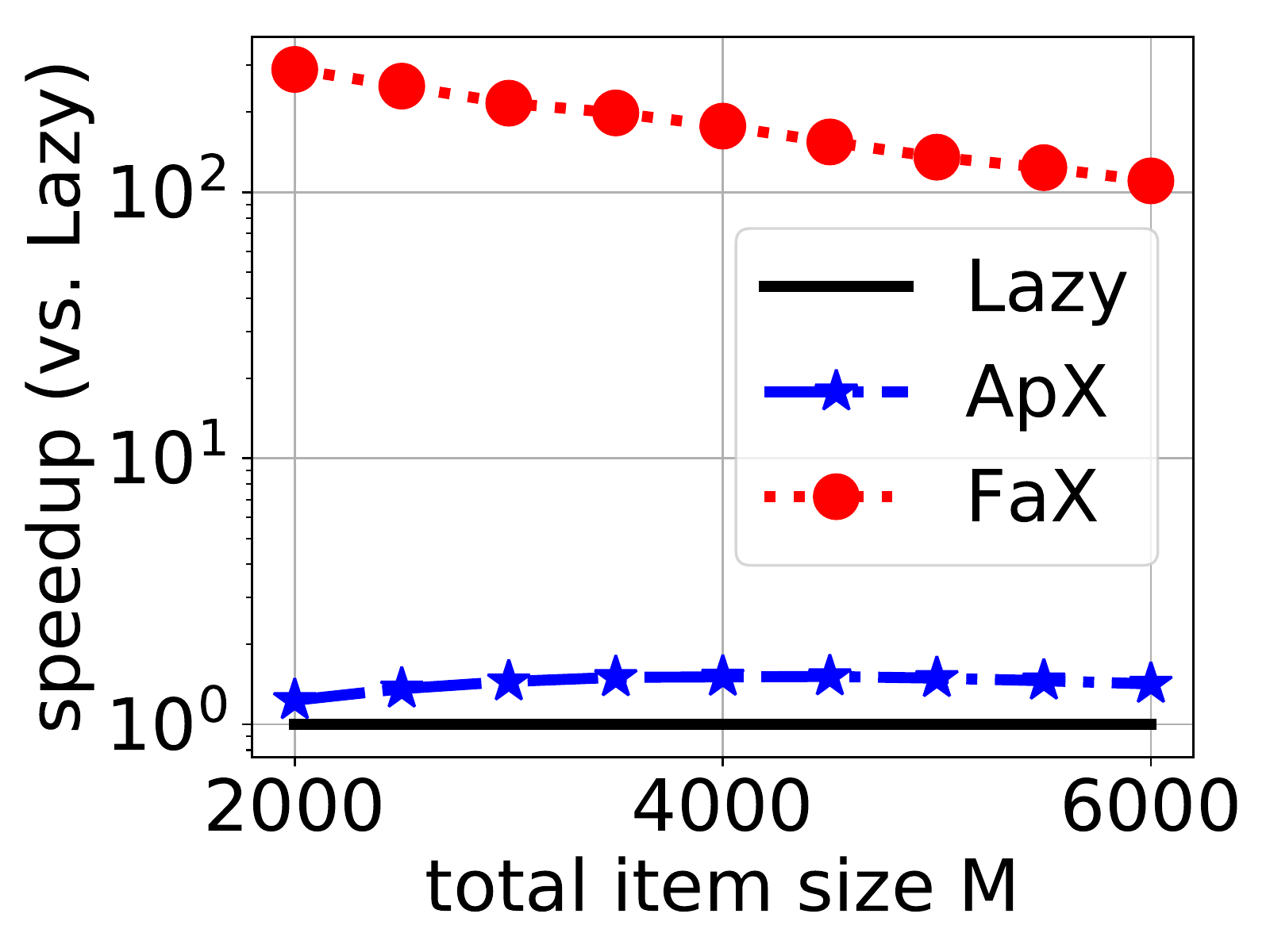}
	\end{minipage}
	\begin{minipage}{.24\textwidth}
		\includegraphics[width=1.35in]{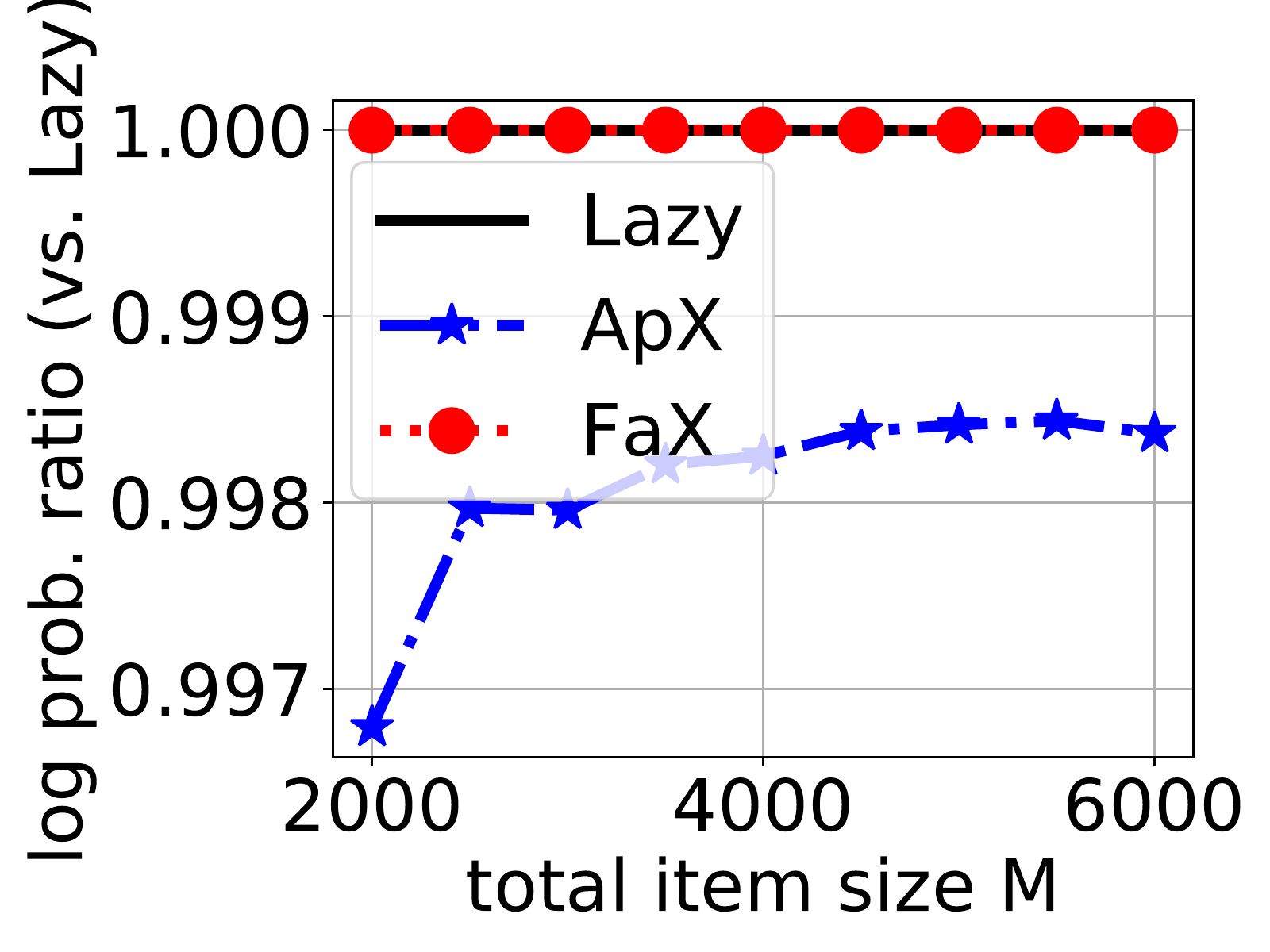}
	\end{minipage}
	\begin{minipage}{.24\textwidth}
		\includegraphics[width=1.35in]{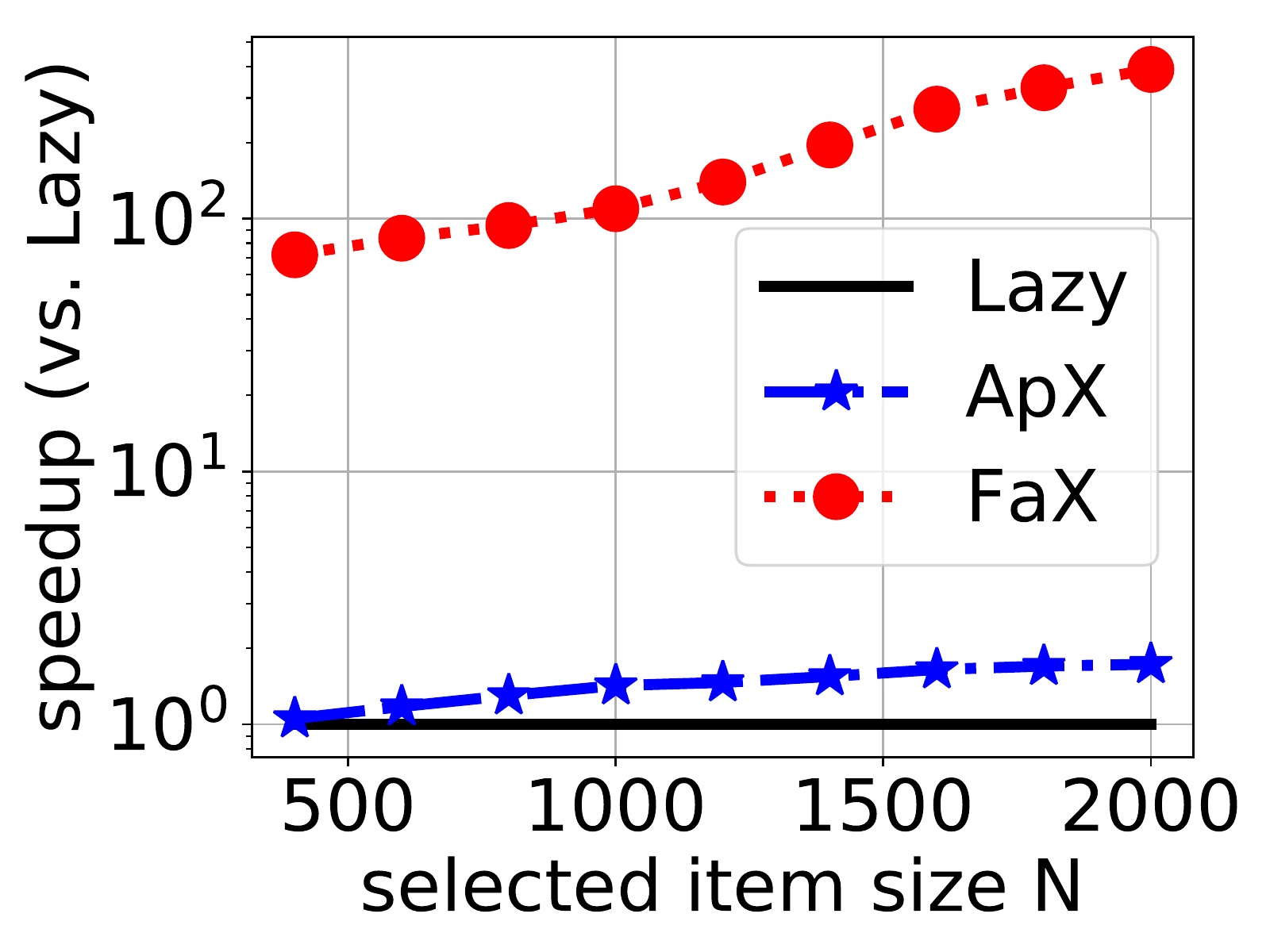}
	\end{minipage}
	\begin{minipage}{.24\textwidth}
		\includegraphics[width=1.35in]{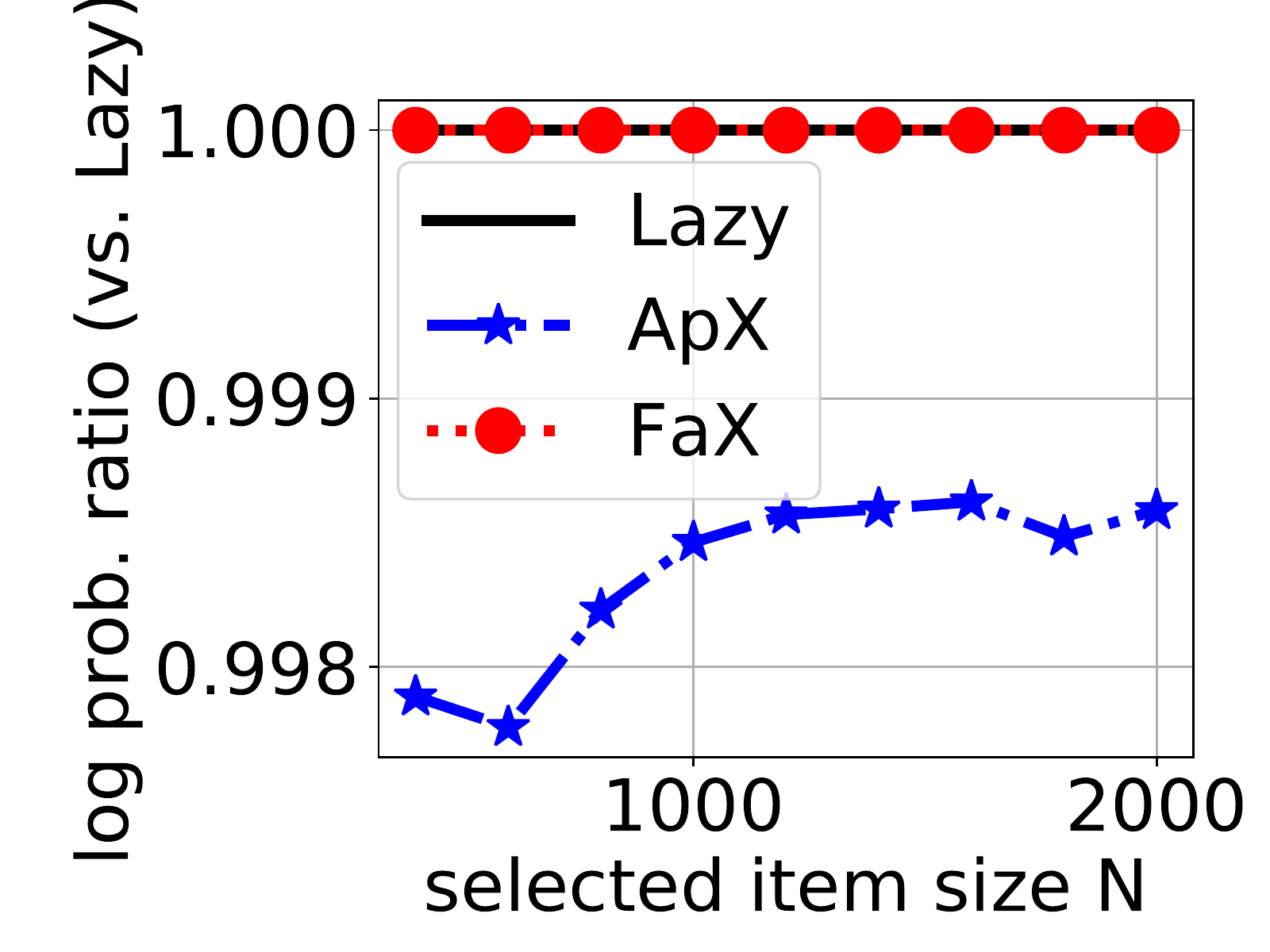}
	\end{minipage}
	\caption{Comparison of Lazy, ApX, and our FaX under different $M$ when $N=1000$ (left), and under different $N$ when $M=6000$ (right).}
	\label{dpp:compare}
\end{center}
\vskip -0.1in
\end{figure}

In this subsection, we evaluate the performance of our Algorithm~\ref{sup:fast} on the MAP inference for DPP.
We follow the experimental setup in \cite{han2017faster}.
The entries of the kernel matrix satisfy Equ. \eqref{kernel:entry},
where $r_i=\exp(0.01x_i+0.2)$ with $x_i\in\mathbb{R}$ drawn from the normal distribution $\mathcal{N}(0,1)$,
and ${\bf f}_i\in\mathbb{R}^D$ with $D$ same as the total item size $M$ and entries drawn i.i.d. from $\mathcal{N}(0,1)$ and then normalized.

Our proposed faster exact algorithm (FaX) is compared with
Schur complement combined with lazy evaluation (Lazy) \cite{minoux1978accelerated}
and faster approximate algorithm (ApX) \cite{han2017faster}.
The parameters of the reference algorithms are chosen as suggested in \cite{han2017faster}.
The gradient-based method in \cite{gillenwater2012near} and the double greedy algorithm in \cite{buchbinder2015tight}
are not compared because as reported in \cite{han2017faster}, they performed worse than ApX.
We report the speedup over Lazy of each algorithm, as well as the ratio of log-probability \cite{han2017faster}
\begin{displaymath}
  \log\det{\bf L}_{Y}/\log\det{\bf L}_{Y_{\textrm{Lazy}}},
\end{displaymath}
where $Y$ and $Y_{\mathrm{Lazy}}$ are the outputs of an algorithm and Lazy, respectively.
We compare these metrics when the total item size $M$ varies from $2000$ to $6000$ with the selected item size $N=1000$,
and when $N$ varies from $400$ to $2000$ with $M=6000$.
The results are averaged over $10$ independent trials, and shown in Figure~\ref{dpp:compare}.
In both cases, FaX runs significantly faster than ApX, which is the state-of-the-art fast greedy MAP inference algorithm for DPP.
FaX is about $100$ times faster than Lazy, while ApX is about $3$ times faster, as reported in \cite{han2017faster}.
The accuracy of FaX is the same as Lazy, because they are exact implementations of the greedy algorithm.
ApX loses about $0.2\%$ accuracy.

\subsection{Short Sequence Recommendation}

In this subsection, we evaluate the performance of Algorithm~\ref{sup:fast}
to recommend short sequences of items to users on the following two public datasets.

\textbf{Netflix Prize}\footnote{Netflix Prize website, http://www.netflixprize.com/}:
This dataset contains users' ratings of movies.
We keep ratings of four or higher and binarize them.
We only keep users who have watched at least $10$ movies and movies that are watched by at least $100$ users.
This results in $436,674$ users and $11,551$ movies with $56,406,404$ ratings.

\textbf{Million Song Dataset} \cite{bertin2011million}:
This dataset contains users' play counts of songs.
We binarize play counts of more than once.
We only keep users who listen to at least $20$ songs and songs that are listened to by at least $100$ users.
This results in $213,949$ users and $20,716$ songs with $8,500,651$ play counts.

For each dataset, we construct the test set by randomly selecting one interacted item for each user, and use the rest data for training.
We adopt an item-based recommendation algorithm \cite{karypis2001evaluation} on the training set
to learn an item-item PSD similarity matrix $\bf S$.
For each user, the profile set $P_u$ consists of the interacted items in the training set,
and the candidate set $C_u$ is formed by the union of $50$ most similar items of each item in $P_u$.
The median of $\#C_u$ is $735$ and $811$ on Netflix Prize and Million Song Dataset, respectively.
For any item in $C_u$, the relevance score is the aggregated similarity to all items in $P_u$ \cite{hurley2011novelty}.
With $\bf S$, $C_u$, and the score vector ${\bf r}_u$, algorithms recommend $N=20$ items.

Performance metrics of recommendation include mean reciprocal rank (MRR) \cite{voorhees1999trec},
intra-list average distance (ILAD) \cite{zhang2008avoiding}, and intra-list minimal distance (ILMD).
They are defined as
\begin{displaymath}
  \mathrm{MRR}=\mean_{u\in U}p_u^{-1},\quad
  \mathrm{ILAD}=\mean_{u\in U}\mean_{i,j\in R_u,i\ne j}(1-{\bf S}_{ij}),\quad
  \mathrm{ILMD}=\mean_{u\in U}\min_{i,j\in R_u,i\ne j}(1-{\bf S}_{ij}),
\end{displaymath}
where $U$ is the set of all users, and $p_u$ is the smallest rank position of items in the test set.
MRR measures relevance, while ILAD and ILMD measure diversity.
We also compare the metric popularity-weighted recall (PW Recall) \cite{steck2011item} in the supplementary material.

Our DPP-based algorithm (DPP) is compared with maximal marginal relevance (MMR) \cite{carbonell1998use},
max-sum diversification (MSD) \cite{borodin2012max},
entropy regularizer (Entropy) \cite{qin2013promoting}, and coverage-based algorithm (Cover) \cite{puthiya2016coverage}.
They all involve a tunable parameter to adjust the trade-off between relevance and diversity.
For Cover, the parameter is $\gamma\in[0,1]$ which defines the saturation function $f(t)=t^{\gamma}$.

\begin{figure}[t]
	\begin{center}
		\begin{minipage}{.24\textwidth}
			\includegraphics[width=1.35in]{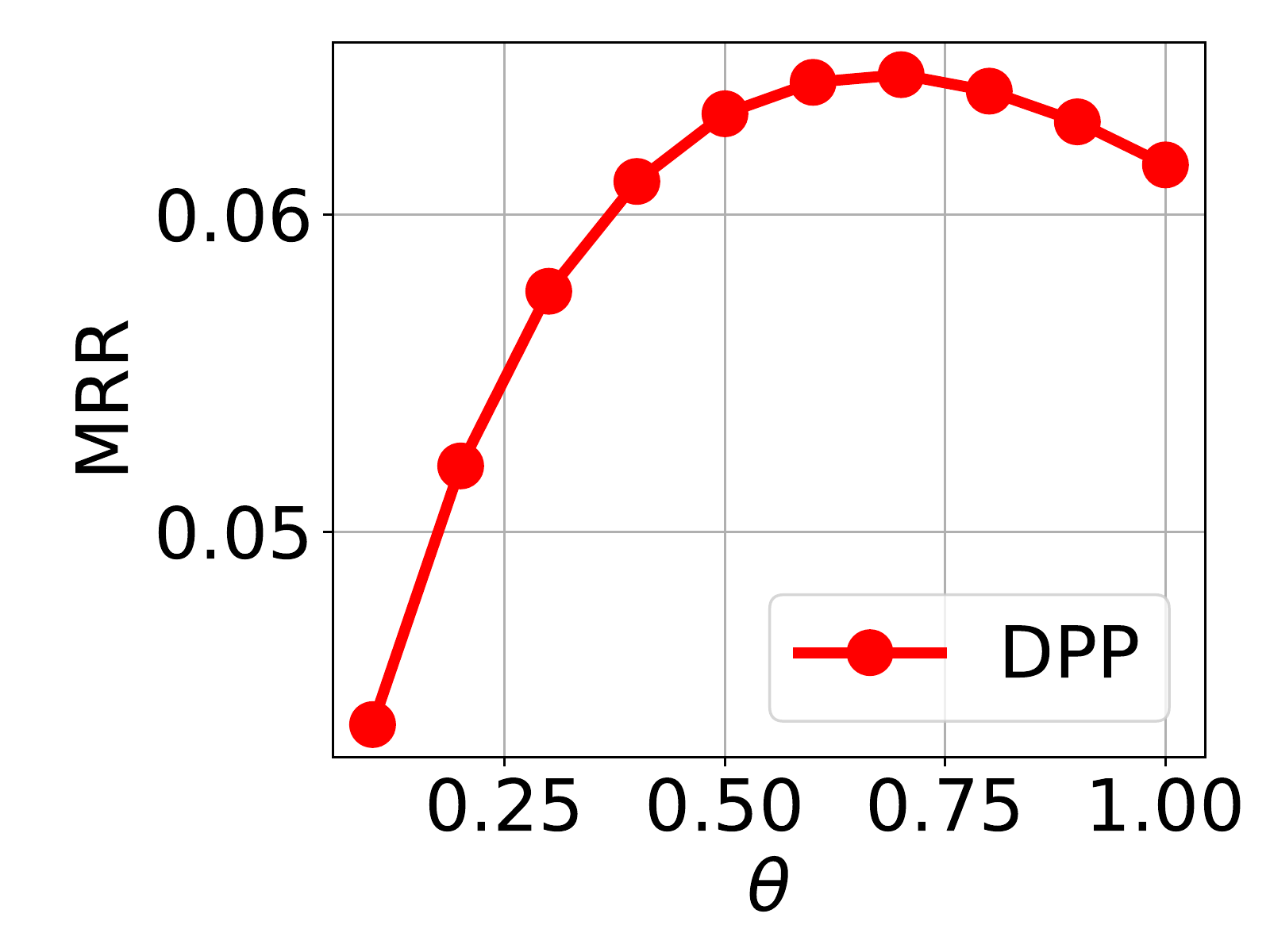}
		\end{minipage}
		\begin{minipage}{.24\textwidth}
			\includegraphics[width=1.35in]{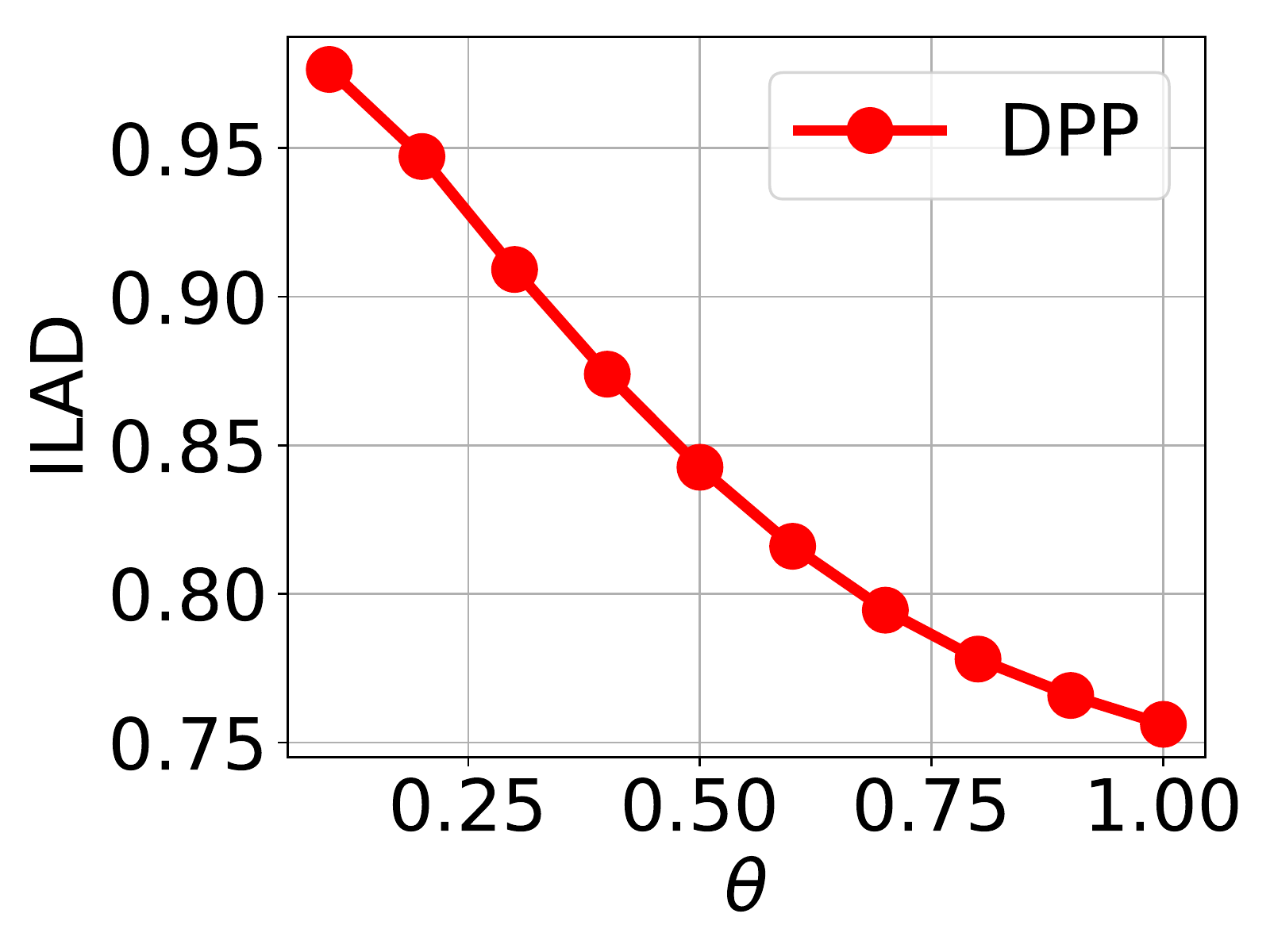}
		\end{minipage}
		\begin{minipage}{.24\textwidth}
			\includegraphics[width=1.35in]{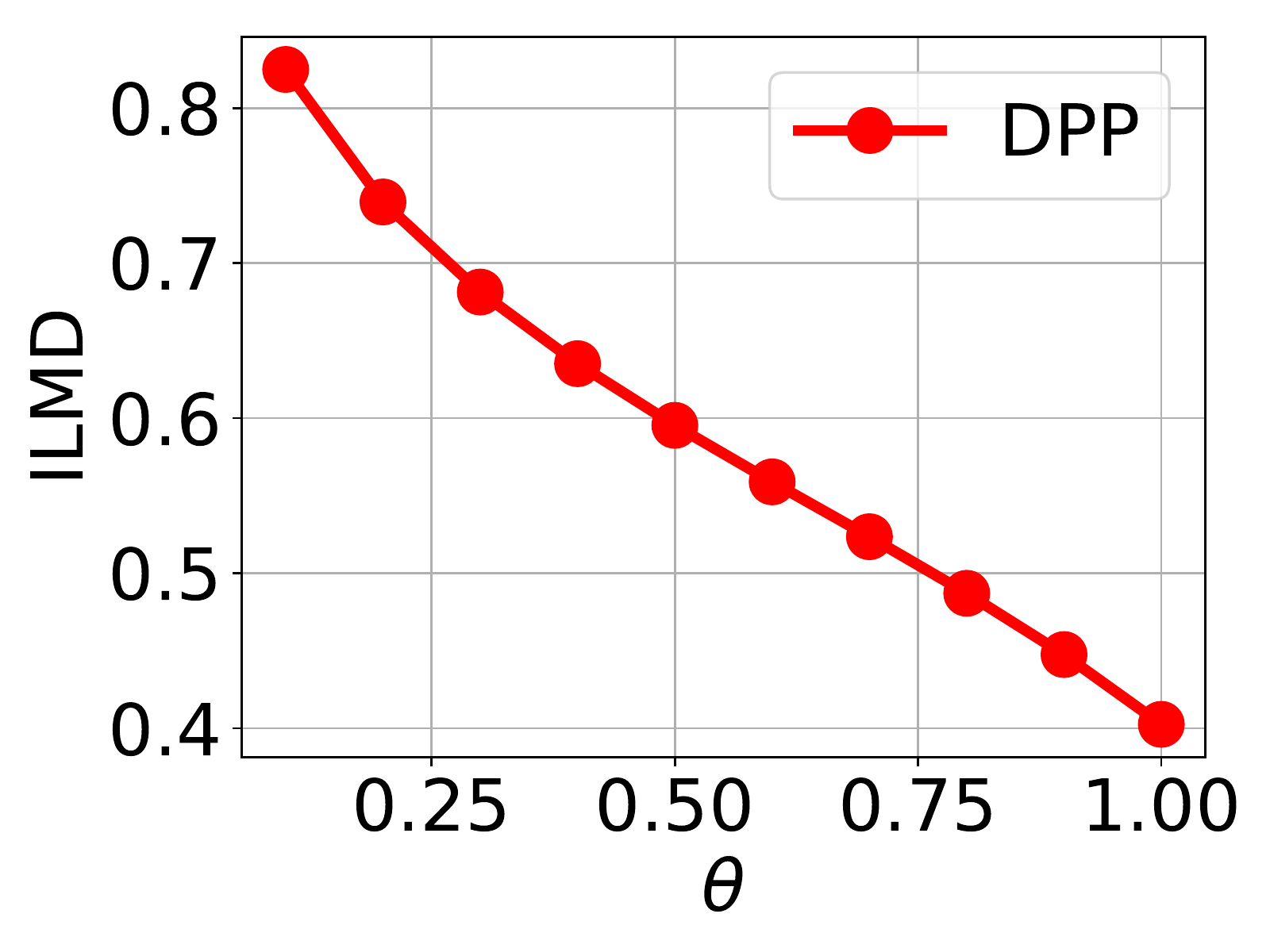}
		\end{minipage}
		\caption{Impact of trade-off parameter $\theta$ on Netflix dataset.}
		\label{theta:compare}
	\end{center}
	\vskip -0.1in
\end{figure}

\begin{figure}[t]
	\begin{center}
		\begin{minipage}{.24\textwidth}
			\includegraphics[width=1.35in]{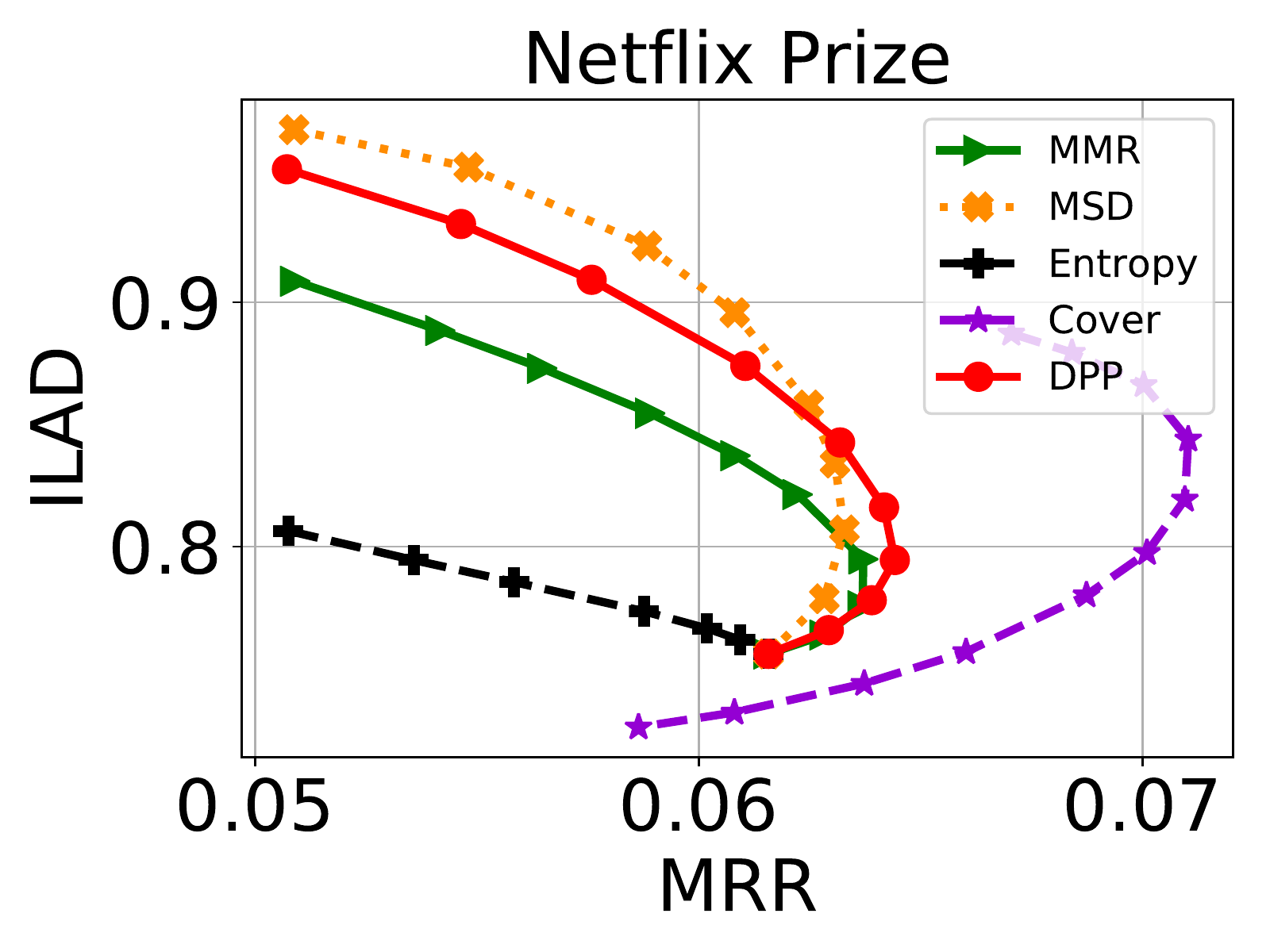}
		\end{minipage}
		\begin{minipage}{.24\textwidth}
			\includegraphics[width=1.35in]{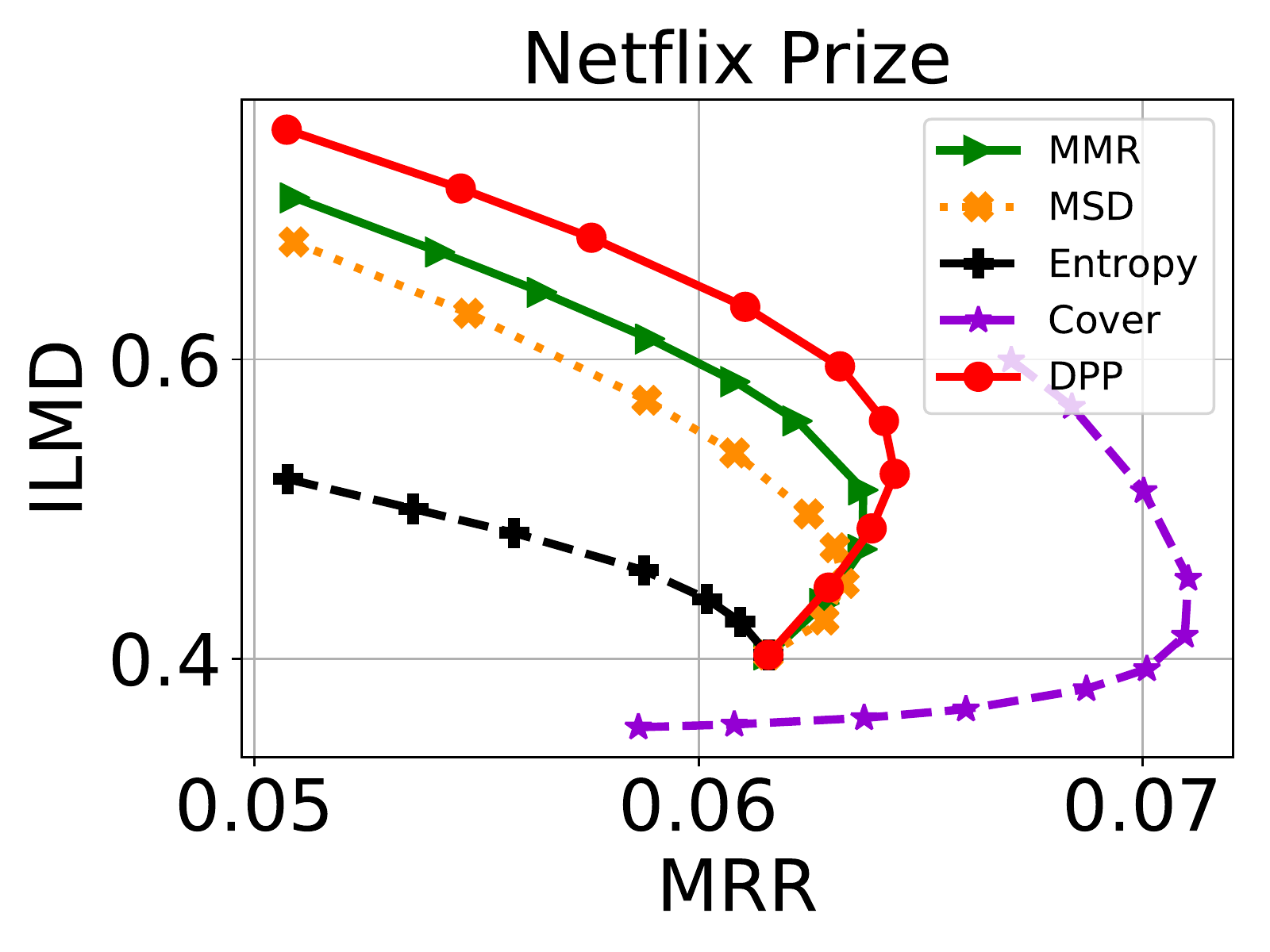}
		\end{minipage}
		\begin{minipage}{.24\textwidth}
			\includegraphics[width=1.35in]{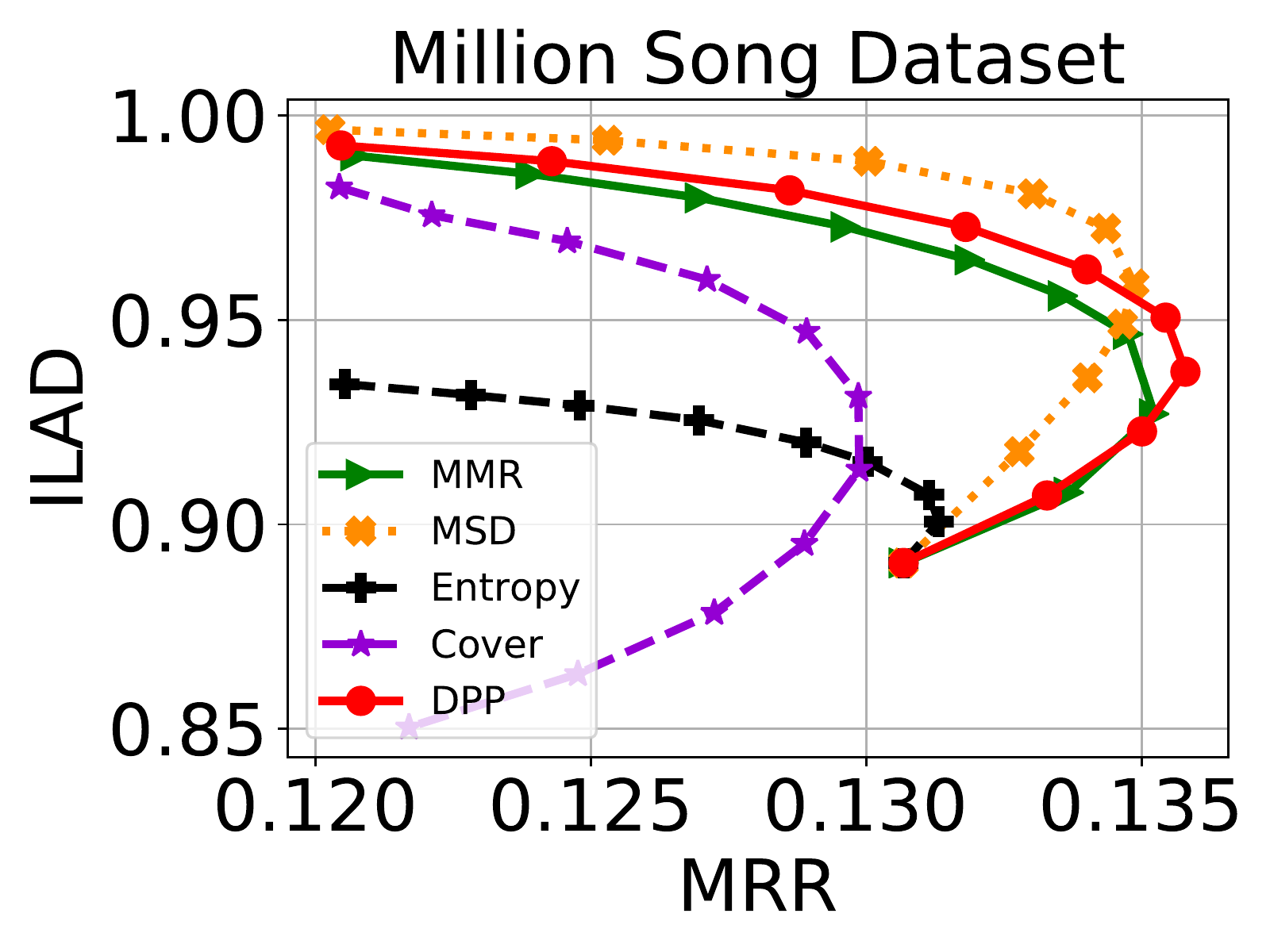}
		\end{minipage}
		\begin{minipage}{.24\textwidth}
			\includegraphics[width=1.35in]{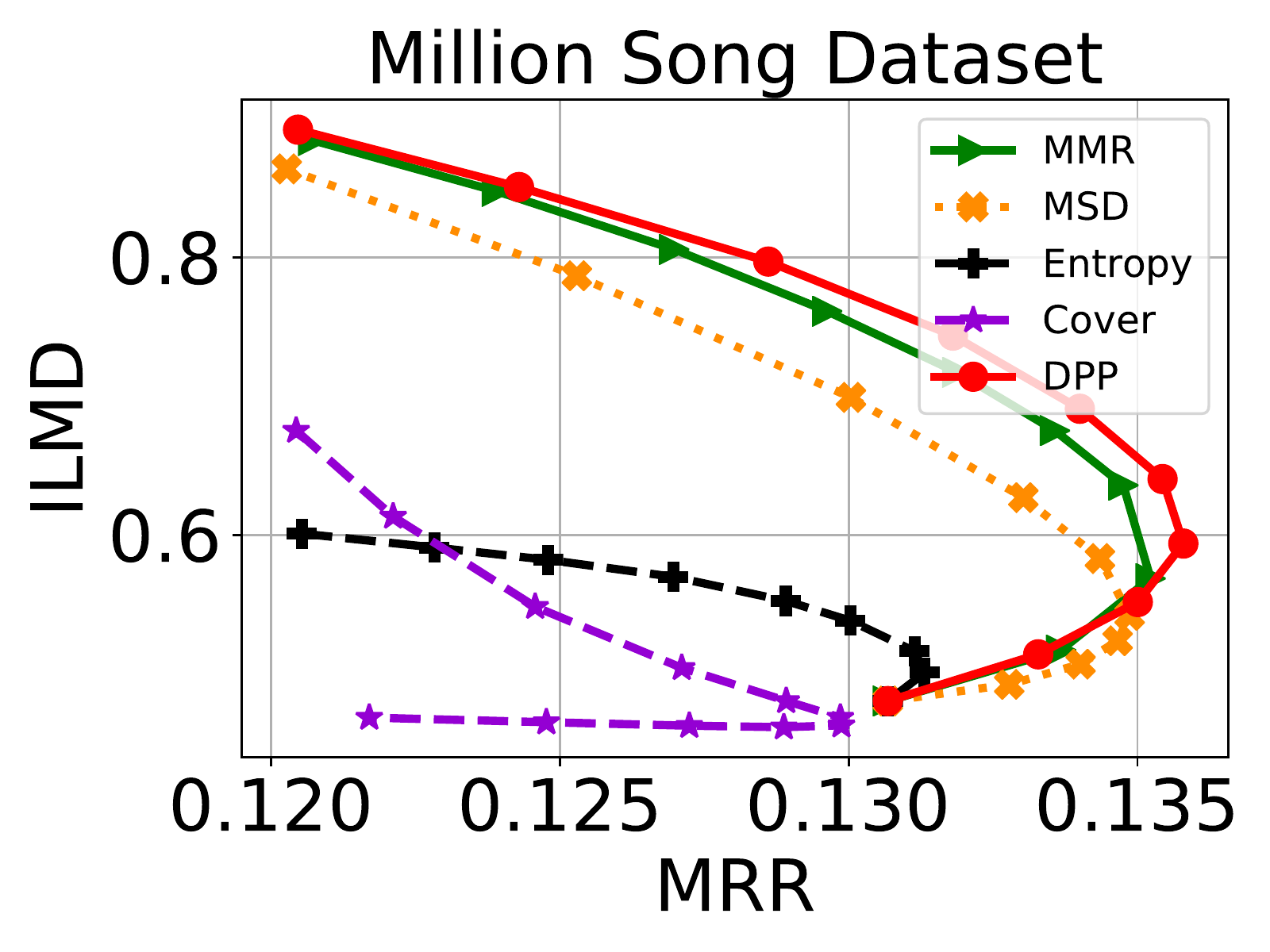}
		\end{minipage}
		\caption{Comparison of trade-off performance between relevance and diversity
			under different choices of trade-off parameters
			on Netflix Prize (left) and Million Song Dataset (right).
			The standard error of MRR is about $0.0003$ and $0.0006$ on Netflix Prize and Million Song Dataset, respectively.}
		\label{one:out}
	\end{center}
	\vskip -0.05in
\end{figure}

\begin{table}[!t]
	\caption{Comparison of average / upper $99\%$ running time (in milliseconds).}
	\label{running:time}
	\begin{center}
		\begin{small}
			\begin{tabular}{cccccc}
				\toprule
				Dataset & MMR  & MSD & Entropy & Cover & DPP \\
				\midrule
				Netflix Prize & 0.23 / 0.50  & 0.21 / 0.41 & 200.74 / 2883.82 & 120.19 / 1332.21 & 0.73 / 1.75 \\
				Million Song Dataset & 0.23 / 0.41  & 0.22 / 0.34 & 26.45 / 168.12 & 23.76 / 173.64 & 0.76 / 1.46 \\
				\bottomrule
			\end{tabular}
		\end{small}
	\end{center}
	\vskip -0.1in
\end{table}

In the first experiment, we test the impact of trade-off parameter $\theta\in[0,1]$ of DPP on Netflix Prize.
The results are shown in Figure~\ref{theta:compare}.
As $\theta$ increases, MRR improves at first, achieves the best value when $\theta\approx0.7$, and then decreases a little bit.
ILAD and ILMD are monotonously decreasing as $\theta$ increases.
When $\theta=1$, DPP returns items with the highest relevance scores.
Therefore, taking moderate amount of diversity into consideration, better recommendation performance can be achieved.

In the second experiment, by varying the trade-off parameters,
the trade-off performance between relevance and diversity are compared in Figure~\ref{one:out}.
The parameters are chosen such that different algorithms have approximately the same range of MRR.
As can be seen, Cover performs the best on Netflix Prize but becomes the worst on Million Song Dataset.
Among the other algorithms, DPP enjoys the best relevance-diversity trade-off performance.
Their average and upper $99\%$ running time are compared in Table~\ref{running:time}.
MMR, MSD, and DPP run significantly faster than Entropy and Cover.
Since DPP runs in less than $2$ms with probability $99\%$, it can be used in real-time scenarios.

\begin{table}[t]
	\vskip 0.05in
	\caption{Performance improvement of MMR and DPP over Control in an online A/B test.}
	\label{ab:test}
	\begin{center}
		\begin{small}
			\begin{tabular}{ccc}
				\toprule
				Algorithm & Improvement of No. Titles Watched & Improvement of Watch Minutes \\
				\midrule
				MMR & $0.84\%$  & $0.86\%$ \\
				DPP & $1.33\%$  & $1.52\%$ \\
				\bottomrule
			\end{tabular}
		\end{small}
	\end{center}
	\vskip -0.05in
\end{table}

We conducted an online A/B test in a movie recommender system for four weeks.
For each user, candidate movies with relevance scores were generated by an online scoring model.
An offline matrix factorization algorithm \cite{koren2009matrix} was trained daily to generate movie representations which were used to get similarities.
For the control group, $5\%$ users were randomly selected and presented with $N=8$ movies with the highest relevance scores.
For the treatment group, another $5\%$ random users were presented with $N$ movies generated by DPP with a fine-tuned trade-off parameter.
Two online metrics, improvements of number of titles watched and watch minutes, are reported in Table~\ref{ab:test}.
The results are also compared with another $5\%$ randomly selected users using MMR.
As can be seen, DPP performed better compared with systems without diversification algorithm or with MMR.

\subsection{Long Sequence Recommendation}

\begin{figure}[!t]
	\begin{center}
		\begin{minipage}{.24\textwidth}
			\includegraphics[width=1.35in]{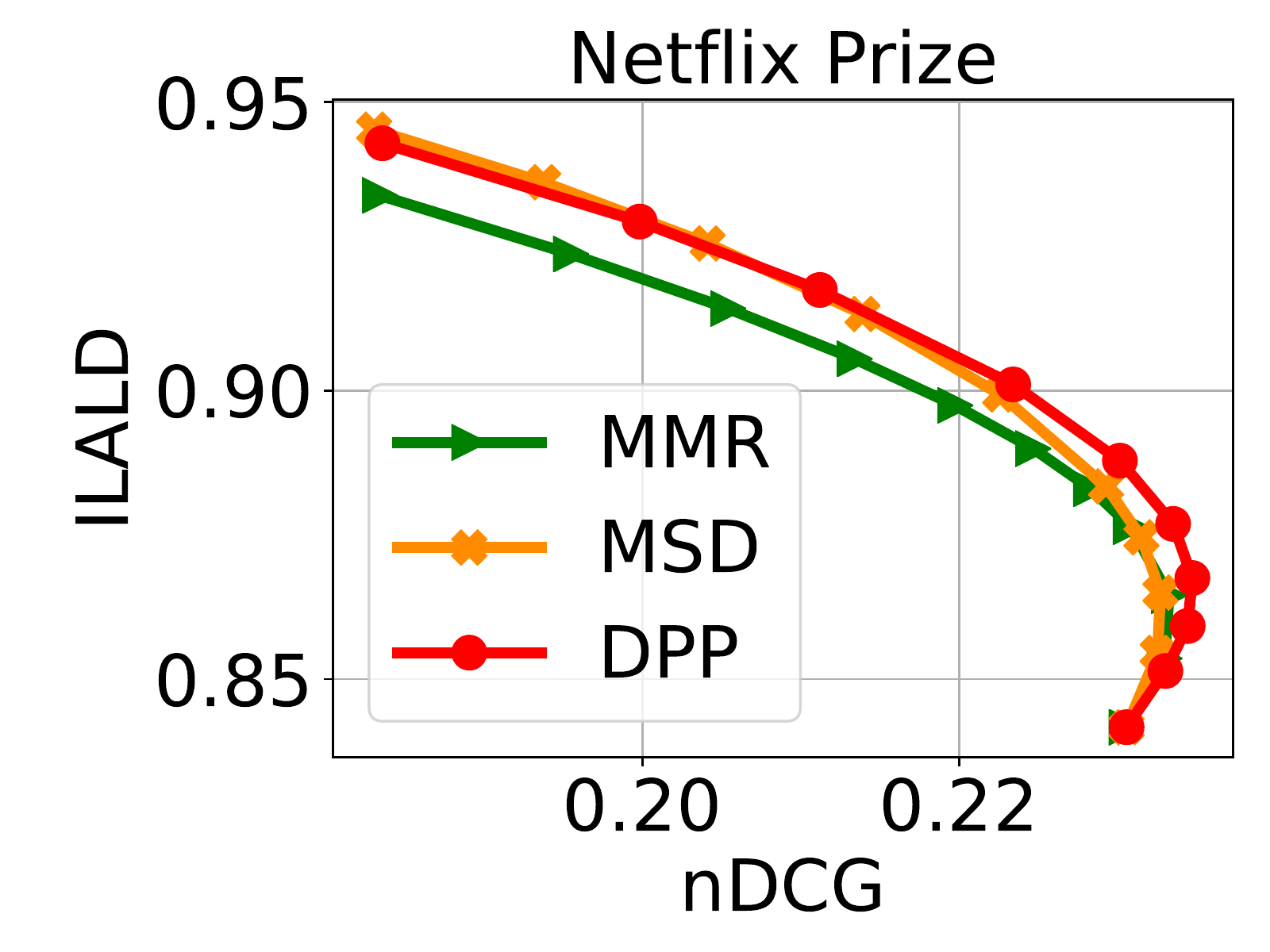}
		\end{minipage}
		\begin{minipage}{.24\textwidth}
			\includegraphics[width=1.35in]{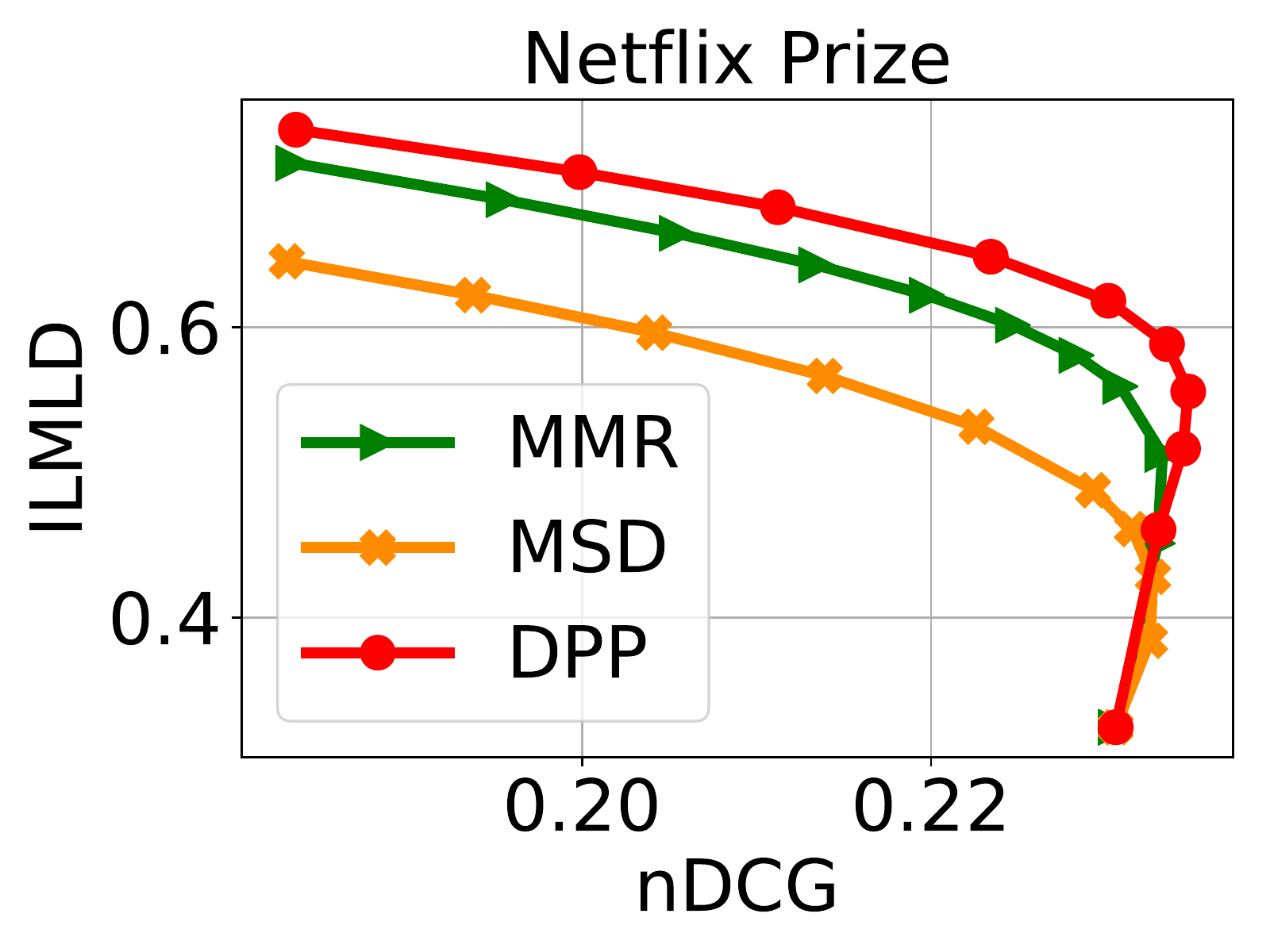}
		\end{minipage}
		\begin{minipage}{.24\textwidth}
			\includegraphics[width=1.35in]{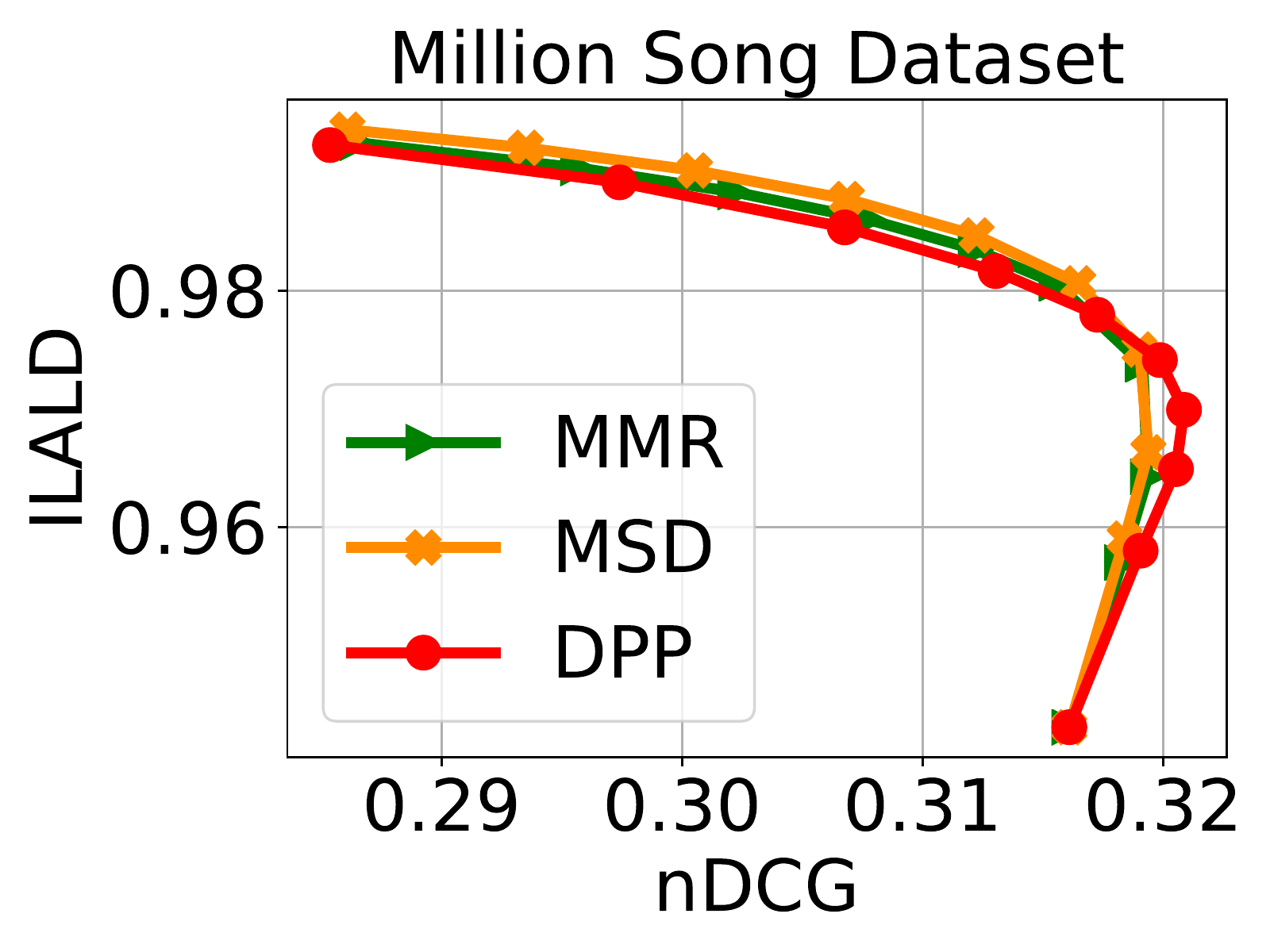}
		\end{minipage}
		\begin{minipage}{.24\textwidth}
			\includegraphics[width=1.35in]{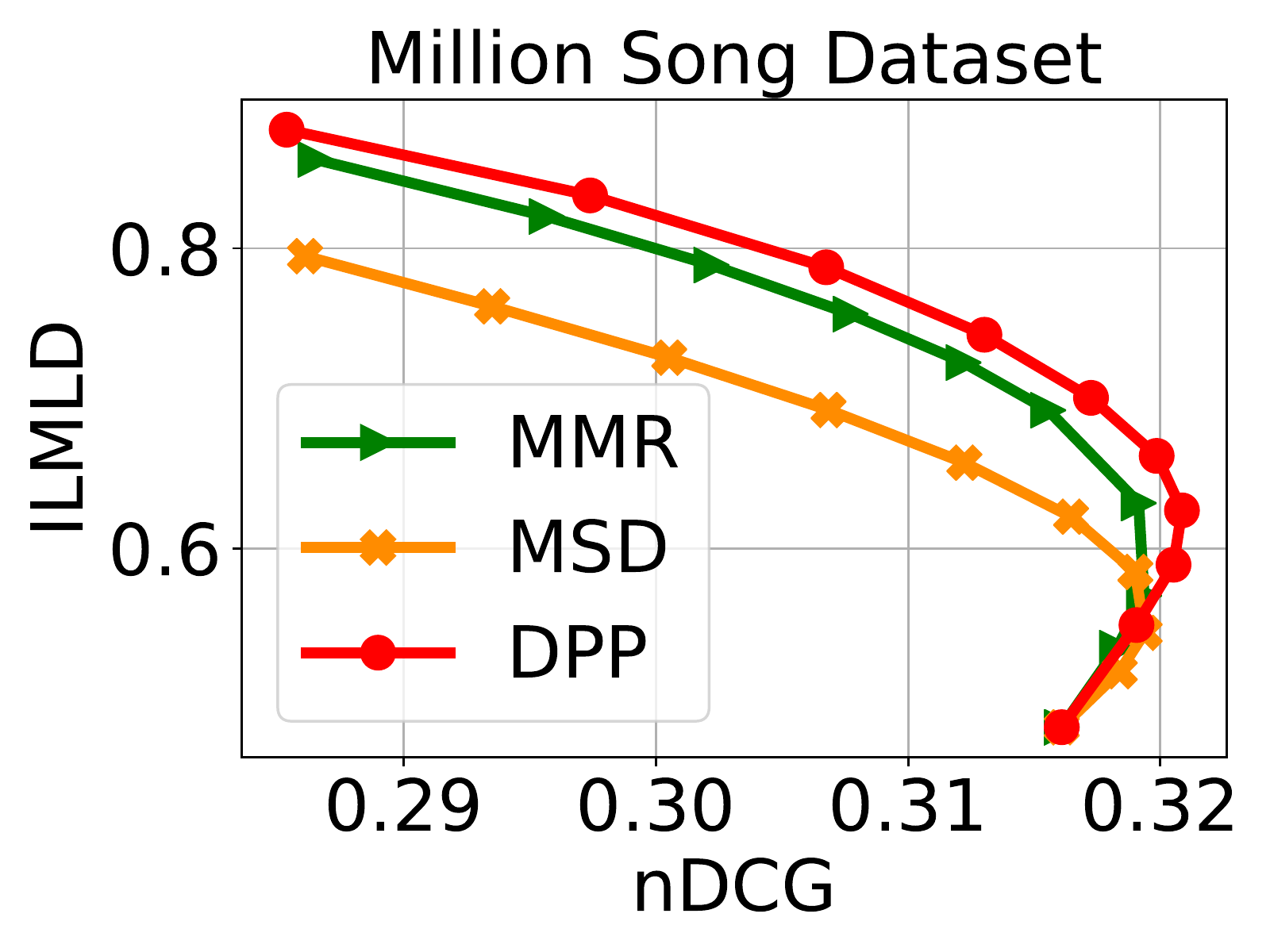}
		\end{minipage}
		\caption{Comparison of trade-off performance between relevance and diversity
			under different choices of trade-off parameters on Netflix Prize (left) and Million Song Dataset (right).
			The standard error of nDCG is about $0.00025$ and $0.0005$ on Netflix Prize and Million Song Dataset, respectively.}
		\label{five:out}
	\end{center}
	\vskip -0.1in
\end{figure}

In this subsection, we evaluate the performance of Algorithm~\ref{diverse:nearby} to recommend long sequences of items to users.
For each dataset, we construct the test set by randomly selecting $5$ interacted items for each user, and use the rest for training.
Each long sequence contains $N=100$ items.
We choose window size $w=10$ so that every $w$ successive items in the sequence are diverse.
Other settings are the same as in the previous subsection.

Performance metrics include normalized discounted cumulative gain (nDCG) \cite{jarvelin2000ir},
intra-list average local distance (ILALD), and intra-list minimal local distance (ILMLD).
The latter two are defined as
\begin{displaymath}
  \mathrm{ILALD}=\mean_{u\in U}\mean_{i,j\in R_u,i\ne j,d_{ij}\le w}(1-{\bf S}_{ij}),\quad
  \mathrm{ILMLD}=\mean_{u\in U}\min_{i,j\in R_u,i\ne j,d_{ij}\le w}(1-{\bf S}_{ij}),
\end{displaymath}
where $d_{ij}$ is the position distance of item $i$ and $j$ in $R_u$.
To make a fair comparison, we modify the diversity terms in MMR and MSD so that they only consider the most recently added $w-1$ items.
Entropy and Cover are not tested because they are not suitable for this scenario.
By varying trade-off parameters, the trade-off performance between relevance and diversity of
MMR, MSD, and DPP are compared in Figure~\ref{five:out}.
The parameters are chosen such that different algorithms have approximately the same range of nDCG.
As can be seen, DPP performs the best with respect to relevance-diversity trade-off.
We also compare the metric PW Recall in the supplementary material.

\section{Conclusion and Future Work}

In this paper, we presented a fast and exact implementation of the greedy MAP inference for DPP.
The time complexity of our algorithm is $O(M^3)$, which is significantly lower than state-of-the-art exact implementations.
Our proposed acceleration technique can be applied to other problems with log-determinant of PSD matrices in the objective functions,
such as the entropy regularizer \cite{qin2013promoting}.
We also adapted our fast algorithm to scenarios where the diversity is only required within a sliding window.
Experiments showed that our algorithm runs significantly faster than state-of-the-art algorithms,
and our proposed approach provides better relevance-diversity trade-off on recommendation task.
A potential future research direction is to learn the optimal trade-off parameter automatically.

\bibliography{reference}

\begin{thebibliography}{10}

\bibitem{adomavicius2011maximizing}
G.~Adomavicius and Y.~Kwon.
\newblock Maximizing aggregate recommendation diversity: A graph-theoretic
  approach.
\newblock In {\em Proceedings of DiveRS 2011}, pages 3--10, 2011.

\bibitem{agrawal2009diversifying}
R.~Agrawal, S.~Gollapudi, A.~Halverson, and S.~Ieong.
\newblock Diversifying search results.
\newblock In {\em Proceedings of WSDM 2009}, pages 5--14. ACM, 2009.

\bibitem{ahmed2012fair}
A.~Ahmed, C.~H. Teo, S.~V.~N. Vishwanathan, and A.~Smola.
\newblock Fair and balanced: Learning to present news stories.
\newblock In {\em Proceedings of WSDM 2012}, pages 333--342. ACM, 2012.

\bibitem{ashkan2015optimal}
A.~Ashkan, B.~Kveton, S.~Berkovsky, and Z.~Wen.
\newblock Optimal greedy diversity for recommendation.
\newblock In {\em Proceedings of IJCAI 2015}, pages 1742--1748, 2015.

\bibitem{aytekin2014clustering}
T.~Aytekin and M.~{\"O}. Karakaya.
\newblock Clustering-based diversity improvement in top-{N} recommendation.
\newblock {\em Journal of Intelligent Information Systems}, 42(1):1--18, 2014.

\bibitem{bertin2011million}
T.~Bertin-Mahieux, D.~P.W. Ellis, B.~Whitman, and P.~Lamere.
\newblock The million song dataset.
\newblock In {\em Proceedings of ISMIR 2011}, page~10, 2011.

\bibitem{boim2011diversification}
R.~Boim, T.~Milo, and S.~Novgorodov.
\newblock Diversification and refinement in collaborative filtering
  recommender.
\newblock In {\em Proceedings of CIKM 2011}, pages 739--744. ACM, 2011.

\bibitem{borodin2012max}
A.~Borodin, H.~C. Lee, and Y.~Ye.
\newblock Max-sum diversification, monotone submodular functions and dynamic
  updates.
\newblock In {\em Proceedings of SIGMOD-SIGACT-SIGAI}, pages 155--166. ACM,
  2012.

\bibitem{bradley2001improving}
K.~Bradley and B.~Smyth.
\newblock Improving recommendation diversity.
\newblock In {\em Proceedings of AICS 2001}, pages 85--94, 2001.

\bibitem{buchbinder2015tight}
N.~Buchbinder, M.~Feldman, J.~Seffi, and R.~Schwartz.
\newblock A tight linear time (1/2)-approximation for unconstrained submodular
  maximization.
\newblock {\em SIAM Journal on Computing}, 44(5):1384--1402, 2015.

\bibitem{carbonell1998use}
J.~Carbonell and J.~Goldstein.
\newblock The use of {MMR}, diversity-based reranking for reordering documents
  and producing summaries.
\newblock In {\em Proceedings of SIGIR 1998}, pages 335--336. ACM, 1998.

\bibitem{chandar2013preference}
P.~Chandar and B.~Carterette.
\newblock Preference based evaluation measures for novelty and diversity.
\newblock In {\em Proceedings of SIGIR 2013}, pages 413--422. ACM, 2013.

\bibitem{ccivril2009selecting}
A.~{\c{C}}ivril and M.~Magdon-Ismail.
\newblock On selecting a maximum volume sub-matrix of a matrix and related
  problems.
\newblock {\em Theoretical Computer Science}, 410(47-49):4801--4811, 2009.

\bibitem{gartrell2016bayesian}
M.~Gartrell, U.~Paquet, and N.~Koenigstein.
\newblock Bayesian low-rank determinantal point processes.
\newblock In {\em Proceedings of RecSys 2016}, pages 349--356. ACM, 2016.

\bibitem{gartrell2017low}
M.~Gartrell, U.~Paquet, and N.~Koenigstein.
\newblock Low-rank factorization of determinantal point processes.
\newblock In {\em Proceedings of AAAI 2017}, pages 1912--1918, 2017.

\bibitem{gillenwater2014approximate}
J.~Gillenwater.
\newblock {\em Approximate inference for determinantal point processes}.
\newblock University of Pennsylvania, 2014.

\bibitem{gillenwater2012near}
J.~Gillenwater, A.~Kulesza, and B.~Taskar.
\newblock Near-optimal {MAP} inference for determinantal point processes.
\newblock In {\em Proceedings of NIPS 2012}, pages 2735--2743, 2012.

\bibitem{gillenwater2014expectation}
J.~A. Gillenwater, A.~Kulesza, E.~Fox, and B.~Taskar.
\newblock Expectation-maximization for learning determinantal point processes.
\newblock In {\em Proceedings of NIPS 2014}, pages 3149--3157, 2014.

\bibitem{gong2014diverse}
B.~Gong, W.~L. Chao, K.~Grauman, and F.~Sha.
\newblock Diverse sequential subset selection for supervised video
  summarization.
\newblock In {\em Proceedings of NIPS 2014}, pages 2069--2077, 2014.

\bibitem{han2017faster}
I.~Han, P.~Kambadur, K.~Park, and J.~Shin.
\newblock Faster greedy {MAP} inference for determinantal point processes.
\newblock In {\em Proceedings of ICML 2017}, pages 1384--1393, 2017.

\bibitem{he2012gender}
J.~He, H.~Tong, Q.~Mei, and B.~Szymanski.
\newblock {GenDeR}: A generic diversified ranking algorithm.
\newblock In {\em Proceedings of NIPS 2012}, pages 1142--1150, 2012.

\bibitem{hurley2011novelty}
N.~Hurley and M.~Zhang.
\newblock Novelty and diversity in top-{N} recommendation--analysis and
  evaluation.
\newblock {\em ACM Transactions on Internet Technology}, 10(4):14, 2011.

\bibitem{jarvelin2000ir}
K.~J{\"a}rvelin and J.~Kek{\"a}l{\"a}inen.
\newblock {IR} evaluation methods for retrieving highly relevant documents.
\newblock In {\em Proceedings of SIGIR 2000}, pages 41--48. ACM, 2000.

\bibitem{karypis2001evaluation}
G.~Karypis.
\newblock Evaluation of item-based top-{N} recommendation algorithms.
\newblock In {\em Proceedings of CIKM 2001}, pages 247--254. ACM, 2001.

\bibitem{ko1995exact}
C.~W. Ko, J.~Lee, and M.~Queyranne.
\newblock An exact algorithm for maximum entropy sampling.
\newblock {\em Operations Research}, 43(4):684--691, 1995.

\bibitem{koren2009matrix}
Y.~Koren, R.~Bell, and C.~Volinsky.
\newblock Matrix factorization techniques for recommender systems.
\newblock {\em Computer}, 42(8), 2009.

\bibitem{kulesza2010structured}
A.~Kulesza and B.~Taskar.
\newblock Structured determinantal point processes.
\newblock In {\em Proceedings of NIPS 2010}, pages 1171--1179, 2010.

\bibitem{kulesza2011k}
A.~Kulesza and B.~Taskar.
\newblock k-{DPP}s: Fixed-size determinantal point processes.
\newblock In {\em Proceedings of ICML 2011}, pages 1193--1200, 2011.

\bibitem{kulesza2011learning}
A.~Kulesza and B.~Taskar.
\newblock Learning determinantal point processes.
\newblock In {\em Proceedings of UAI 2011}, pages 419--427. AUAI Press, 2011.

\bibitem{kulesza2012determinantal}
A.~Kulesza and B.~Taskar.
\newblock Determinantal point processes for machine learning.
\newblock {\em Foundations and Trends{\textregistered} in Machine Learning},
  5(2--3):123--286, 2012.

\bibitem{lee2017single}
S.~C. Lee, S.~W. Kim, S.~Park, and D.~K. Chae.
\newblock A single-step approach to recommendation diversification.
\newblock In {\em Proceedings of WWW 2017 Companion}, pages 809--810. ACM,
  2017.

\bibitem{li2016gaussian}
C.~Li, S.~Sra, and S.~Jegelka.
\newblock Gaussian quadrature for matrix inverse forms with applications.
\newblock In {\em Proceedings of ICML 2016}, pages 1766--1775, 2016.

\bibitem{macchi1975coincidence}
O.~Macchi.
\newblock The coincidence approach to stochastic point processes.
\newblock {\em Advances in Applied Probability}, 7(1):83--122, 1975.

\bibitem{mariet2015fixed}
Z.~Mariet and S.~Sra.
\newblock Fixed-point algorithms for learning determinantal point processes.
\newblock In {\em Proceedings of ICML 2015}, pages 2389--2397, 2015.

\bibitem{mehta1960density}
M.~L. Mehta and M.~Gaudin.
\newblock On the density of eigenvalues of a random matrix.
\newblock {\em Nuclear Physics}, 18:420--427, 1960.

\bibitem{minoux1978accelerated}
M.~Minoux.
\newblock Accelerated greedy algorithms for maximizing submodular set
  functions.
\newblock In {\em Optimization techniques}, pages 234--243. Springer, 1978.

\bibitem{nemhauser1978analysis}
G.~L. Nemhauser, L.~A. Wolsey, and M.~L. Fisher.
\newblock An analysis of approximations for maximizing submodular set
  functions--{I}.
\newblock {\em Mathematical Programming}, 14(1):265--294, 1978.

\bibitem{niemann2013new}
K.~Niemann and M.~Wolpers.
\newblock A new collaborative filtering approach for increasing the aggregate
  diversity of recommender systems.
\newblock In {\em Proceedings of SIGKDD 2013}, pages 955--963. ACM, 2013.

\bibitem{puthiya2016coverage}
S.~A. {Puthiya Parambath}, N.~Usunier, and Y.~Grandvalet.
\newblock A coverage-based approach to recommendation diversity on similarity
  graph.
\newblock In {\em Proceedings of RecSys 2016}, pages 15--22. ACM, 2016.

\bibitem{qin2013promoting}
L.~Qin and X.~Zhu.
\newblock Promoting diversity in recommendation by entropy regularizer.
\newblock In {\em Proceedings of IJCAI 2013}, pages 2698--2704, 2013.

\bibitem{schott1999matrix}
J.~R. Schott.
\newblock Matrix algorithms, volume 1: Basic decompositions.
\newblock {\em Journal of the American Statistical Association},
  94(448):1388--1388, 1999.

\bibitem{steck2011item}
H.~Steck.
\newblock Item popularity and recommendation accuracy.
\newblock In {\em Proceedings of RecSys 2011}, pages 125--132. ACM, 2011.

\bibitem{su2013set}
R.~Su, L.~A. Yin, K.~Chen, and Y.~Yu.
\newblock Set-oriented personalized ranking for diversified top-{N}
  recommendation.
\newblock In {\em Proceedings of RecSys 2013}, pages 415--418. ACM, 2013.

\bibitem{teo2016adaptive}
C.~H. Teo, H.~Nassif, D.~Hill, S.~Srinivasan, M.~Goodman, V.~Mohan, and
  S.~V.~N. Vishwanathan.
\newblock Adaptive, personalized diversity for visual discovery.
\newblock In {\em Proceedings of RecSys 2016}, pages 35--38. ACM, 2016.

\bibitem{vargas2014coverage}
S.~Vargas, L.~Baltrunas, A.~Karatzoglou, and P.~Castells.
\newblock Coverage, redundancy and size-awareness in genre diversity for
  recommender systems.
\newblock In {\em Proceedings of RecSys 2014}, pages 209--216. ACM, 2014.

\bibitem{voorhees1999trec}
E.~M. Voorhees.
\newblock The {TREC}-8 question answering track report.
\newblock In {\em Proceedings of TREC 1999}, pages 77--82, 1999.

\bibitem{wu2016relevance}
L.~Wu, Q.~Liu, E.~Chen, N.~J. Yuan, G.~Guo, and X.~Xie.
\newblock Relevance meets coverage: A unified framework to generate diversified
  recommendations.
\newblock {\em ACM Transactions on Intelligent Systems and Technology},
  7(3):39, 2016.

\bibitem{xia2017adapting}
L.~Xia, J.~Xu, Y.~Lan, J.~Guo, W.~Zeng, and X.~Cheng.
\newblock Adapting {Markov} decision process for search result diversification.
\newblock In {\em Proceedings of SIGIR 2017}, pages 535--544. ACM, 2017.

\bibitem{yao2016tweet}
J.~G. Yao, F.~Fan, W.~X. Zhao, X.~Wan, E.~Y. Chang, and J.~Xiao.
\newblock Tweet timeline generation with determinantal point processes.
\newblock In {\em Proceedings of AAAI 2016}, pages 3080--3086, 2016.

\bibitem{yu2009takes}
C.~Yu, L.~Lakshmanan, and S.~Amer-Yahia.
\newblock It takes variety to make a world: Diversification in recommender
  systems.
\newblock In {\em Proceedings of EDBT 2009}, pages 368--378. ACM, 2009.

\bibitem{zhang2008avoiding}
M.~Zhang and N.~Hurley.
\newblock Avoiding monotony: Improving the diversity of recommendation lists.
\newblock In {\em Proceedings of RecSys 2008}, pages 123--130. ACM, 2008.

\bibitem{ziegler2005improving}
C.~N. Ziegler, S.~M. McNee, J.~A. Konstan, and G.~Lausen.
\newblock Improving recommendation lists through topic diversification.
\newblock In {\em Proceedings of WWW 2005}, pages 22--32. ACM, 2005.

\end{thebibliography}
\bibliographystyle{plain}

\newpage
\pagebreak

\appendix

\begin{center}
	{\LARGE Supplementary Material}
\end{center}

\vskip 0.1in

\section{Update of $\bf V$, ${\bf c}_i$, and $d_i$}

Hereinafter, we use $Y_{\mathrm{g}}^w$ to denote the set after the earliest added item $h$ is removed.
We will show how to update $\bf V$, ${\bf c}_i$, and $d_i$ when $h$ is removed from $\{h\}\cup Y_{\mathrm{g}}^w$.

\subsection{Update of ${\bf V}$}

Since the Cholesky factor of ${\bf L}_{\{h\}\cup Y_{\mathrm{g}}^w}$ is $\bf V$,
\begin{equation}\label{sub:mat}
  \begin{bmatrix}
    {\bf L}_{hh} & {\bf L}_{h,Y_{\mathrm{g}}^w} \\
    {\bf L}_{Y_{\mathrm{g}}^w,h} & {\bf L}_{Y_{\mathrm{g}}^w}
  \end{bmatrix}
  =\begin{bmatrix}
    {\bf V}_{11} & {\bf 0} \\
    {\bf V}_{2:,1} & {\bf V}_{2:}
  \end{bmatrix}
  \begin{bmatrix}
    {\bf V}_{11} & {\bf 0} \\
    {\bf V}_{2:,1} & {\bf V}_{2:}
  \end{bmatrix}^{\top},
\end{equation}
where $2\!:$ denotes the indices starting from $2$ to the end.
For simplicity, hereinafter we use ${\bf v}={\bf V}_{2:,1}$ and $\ddot{\bf V}={\bf V}_{2:}$.
Let ${\bf V}'$ denote the Cholesky factor of ${\bf L}_{Y_{\mathrm{g}}^w}$.
Equ. \eqref{sub:mat} implies
\begin{equation}\label{V:new}
  {\bf V}'{\bf V}'^{\top}={\bf L}_{Y_{\mathrm{g}}^w}=\ddot{\bf V}\ddot{\bf V}^{\top}+{\bf v}{\bf v}^{\top}.
\end{equation}
Since ${\bf V}'$ and $\ddot{\bf V}$ are lower triangular matrices of the same size, Equ. \eqref{V:new} is a classical rank-one update for Cholesky decomposition.
We follow the procedure described in \cite{schott1999matrix} for this problem.
Let
\begin{displaymath}
  {\bf V}'=\begin{bmatrix}
    {\bf V}'_{11} & {\bf 0} \\
    {\bf V}'_{2:,1} & {\bf V}'_{2:}
  \end{bmatrix},\,
  \ddot{\bf V}=\begin{bmatrix}
    \ddot{\bf V}_{11} & {\bf 0} \\
    \ddot{\bf V}_{2:,1} & \ddot{\bf V}_{2:}
  \end{bmatrix},\,
  {\bf v}=\begin{bmatrix}
    {\bf v}_1 \\
    {\bf v}_{2:}
  \end{bmatrix}.
\end{displaymath}
Then Equ. \eqref{V:new} is equivalent to
\begin{align}
  {\bf V}'^2_{11}&=\ddot{\bf V}_{11}^2+{\bf v}_1^2,\label{V11} \\
  {\bf V}'_{2:,1}{\bf V}'_{11}&=\ddot{\bf V}_{2:,1}\ddot{\bf V}_{11}+{\bf v}_{2:}{\bf v}_1,\nonumber \\
  {\bf V}'_{2:}{\bf V}'^{\top}_{2:}+{\bf V}'_{2:,1}{\bf V}'^{\top}_{2:,1}&=\ddot{\bf V}_{2:}\ddot{\bf V}_{2:}^{\top}+\ddot{\bf V}_{2:,1}\ddot{\bf V}_{2:,1}^{\top}+{\bf v}_{2:}{\bf v}_{2:}^{\top}\nonumber
\end{align}
which imply
\begin{align}
  {\bf V}'_{2:,1}&=(\ddot{\bf V}_{2:,1}\ddot{\bf V}_{11}+{\bf v}_{2:}{\bf v}_1)/{\bf V}'_{11},\label{V1:down} \\
  {\bf V}'_{2:}{\bf V}'^{\top}_{2:}&=\ddot{\bf V}_{2:}\ddot{\bf V}_{2:}^{\top}+{\bf v}'{\bf v}'^{\top},\label{V:right} \\
  {\bf v}'&=({\bf v}_{2:}{\bf V}'_{11}-{\bf V}'_{2:,1}{\bf v}_1)/\ddot{\bf V}_{11}.\label{v:re}
\end{align}
The first column of ${\bf V}'$ can be determined by Equ. \eqref{V11} and Equ. \eqref{V1:down}.
For the rest part, notice that Equ. \eqref{V:right} together with Equ. \eqref{v:re} is again a rank-one update but with a smaller size.
We can repeat the aforementioned procedure until the last diagonal element of ${\bf V}'$ is obtained.

\subsection{Update of ${\bf c}_i$}

According to Equ. \eqref{ci:direct},
\begin{displaymath}
  \begin{bmatrix}
    {\bf V}_{11} & {\bf 0} \\
    {\bf v} & \ddot{\bf V}
  \end{bmatrix}
  \begin{bmatrix}
    {\bf c}_{i,1} & {\bf c}_{i,2:}
  \end{bmatrix}^{\top}
  =\begin{bmatrix}
    {\bf L}_{hi} \\
    {\bf L}_{Y_{\mathrm{g}}^w,i}
  \end{bmatrix},
\end{displaymath}
where ${\bf c}_{i,1}$ denotes the first element of ${\bf c}_i$ and ${\bf c}_{i,2:}$ is the remaining sub-vector.
Let $a_i={\bf c}_{i,1}$ and $\ddot{\bf c}_i={\bf c}_{i,2:}$.
Define ${\bf c}'_i$ as the vector of item $i$ after $h$ is removed.
Then
\begin{equation}\label{ci:new}
  {\bf V}'{\bf c}'^{\top}_i={\bf L}_{Y_{\mathrm{g}}^w,i}=\ddot{\bf V}\ddot{\bf c}_i^{\top}+{\bf v}a_i.
\end{equation}
Let
\begin{displaymath}
  {\bf c}'_i=\begin{bmatrix} {\bf c}'_{i,1} & {\bf c}'_{i,2:} \end{bmatrix},\,
  \ddot{\bf c}_i=\begin{bmatrix} \ddot{\bf c}_{i,1} & \ddot{\bf c}_{i,2:} \end{bmatrix}.
\end{displaymath}
Then Equ. \eqref{ci:new} is equivalent to
\begin{align*}
  {\bf c}'_{i,1}{\bf V}'_{11}&=\ddot{\bf c}_{i,1}\ddot{\bf V}_{11}+a_{i}{\bf v}_1, \\
  {\bf V}'_{2:}{\bf c}'^{\top}_{i,2:}+{\bf V}'_{2:,1}{\bf c}'_{i,1}&=\ddot{\bf V}_{2:}\ddot{\bf c}_{i,2:}^{\top}+\ddot{\bf V}_{2:,1}\ddot{\bf c}_{i,1}+{\bf v}_{2:}a_i
\end{align*}
which imply
\begin{align}
  {\bf c}'_{i,1}&=(\ddot{\bf c}_{i,1}\ddot{\bf V}_{11}+a_{i}{\bf v}_1)/{\bf V}'_{11},\label{ci:first} \\
  {\bf V}'_{2:}{\bf c}'^{\top}_{i,2:}&=\ddot{\bf V}_{2:}\ddot{\bf c}_{i,2:}^{\top}+{\bf v}'a'_i,\label{ci:down} \\
  a'_i&=(a_i{\bf V}'_{11}-{\bf c}'_{i,1}{\bf v}_1)/\ddot{\bf V}_{11}.\label{a:re}
\end{align}
The first element of ${\bf c}'_i$ can be determined by Equ. \eqref{ci:first}.
For the rest part, since Equ. \eqref{ci:down} together with Equ. \eqref{v:re} and \eqref{a:re} has the same form as Equ. \eqref{ci:new},
we can repeat the aforementioned procedure until we get the last element of ${\bf c}'_i$.

\subsection{Update of $d_i$}
\label{di:update:nearby}

According to Equ. \eqref{di:direct},
\begin{displaymath}
  d_i^2={\bf L}_{ii}-a_i^2-\|\ddot{\bf c}_i\|_2^2.
\end{displaymath}
Define ${d}'_i$ as the scalar of item $i$ after $h$ is removed.
Then
\begin{align}
  d_i^{\prime2}&={\bf L}_{ii}-\|{\bf c}'_i\|_2^2\nonumber \\
  &=d_i^2+a_i^2+\|\ddot{\bf c}_i\|_2^2-\|{\bf c}'_i\|_2^2\label{di:new} \\
  &=d_i^2+a_i^2+\ddot{\bf c}_{i,1}^2+\|\ddot{\bf c}_{i,2:}\|_2^2-{\bf c}^{\prime2}_{i,1}-\|{\bf c}'_{i,2:}\|_2^2\nonumber \\
  &=d_i^2+a^{\prime2}_i+\|\ddot{\bf c}_{i,2:}\|_2^2-\|{\bf c}'_{i,2:}\|_2^2\label{di:new:new}
\end{align}
where Equ. \eqref{di:new:new} is due to
\begin{align*}
  a_i^2+\ddot{\bf c}_{i,1}^2-{\bf c}^{\prime2}_{i,1}
  &\stackrel{\mathrm{Equ}. \eqref{ci:first}}{=\joinrel=\joinrel=\joinrel=}a_i^2+({\bf c}'_{i,1}{\bf V}'_{11}-a_{i}{\bf v}_1)^2/\ddot{\bf V}_{11}^2-{\bf c}^{\prime2}_{i,1} \\
  &=(a_i^2(\ddot{\bf V}_{11}^2+{\bf v}_1^2)+{\bf c}^{\prime2}_{i,1}({\bf V}_{11}^{\prime2}-\ddot{\bf V}_{11}^2)-2{\bf c}'_{i,1}{\bf V}'_{11}a_{i}{\bf v}_1)/\ddot{\bf V}_{11}^2 \\
  &\stackrel{\mathrm{Equ}. \eqref{V11}}{=\joinrel=\joinrel=\joinrel=}(a_i^2{\bf V}_{11}^{\prime2}+{\bf c}^{\prime2}_{i,1}{\bf v}_1^2-2{\bf c}'_{i,1}{\bf V}'_{11}a_{i}{\bf v}_1)/\ddot{\bf V}_{11}^2 \\
  &\stackrel{\mathrm{Equ}. \eqref{a:re}}{=\joinrel=\joinrel=\joinrel=}a^{\prime2}_i.
\end{align*}
Notice that Equ. \eqref{di:new:new} has the same form as Equ. \eqref{di:new}.
Therefore, after ${\bf c}'_i$ has been updated, we can directly get $d_i^{\prime2}=d_i^2+a'^2_i$.

\section{Discussion on Numerical Stability}

As introduced in Section~\ref{sec:fast}, updating ${\bf c}_i$ and $d_i$ involves calculating $e_i$, where $d_j^{-1}$ is involved.
If $d_j$ is approximately zero, our algorithm encounters the numerical instability issue.
According to Equ. \eqref{di:intep}, $d_j$ satisfies
\begin{equation}\label{dj:intep}
  d_j^2=\frac{\mathrm{det}({\bf L}_{Y_{\mathrm{g}}\cup\{j\}})}{\mathrm{det}({\bf L}_{Y_{\mathrm{g}}})}.
\end{equation}
Let $j_k$ be the selected item in the $k$-th iteration.
Theorem~\ref{thm:dj} gives some results about the sequence $\{d_{j_k}\}$.

\begin{theorem}\label{thm:dj}
  In Algorithm~\ref{sup:fast}, $\{d_{j_k}\}$ is non-increasing, and $d_{j_k}>0$ if and only if $k\le\mathrm{rank}({\bf L})$.
\end{theorem}

\begin{proof}
First, in the $k$-th iteration, since $j_k$ is the solution to Opt.~\eqref{dpp:fast}, $d_{j_k}\ge d_{j_{k+1}}$.
After $j_k$ is added, $d_{j_{k+1}}$ does not increase after update Equ.~\eqref{di:update}.
Therefore, sequence $\{d_{j_k}\}$ is non-increasing.

Now we prove the second part of the theorem.
Let $V_k\subseteq Z$ be the items that have been selected by Algorithm~\ref{sup:fast} at the end of the $k$-th step.
Let $W_k\subseteq Z$ be a set of $k$ items such that $\mathrm{det}({\bf L}_{W_k})$ is maximum.
According to Theorem 3.3 in \cite{ccivril2009selecting}, we have
\begin{displaymath}
  \mathrm{det}({\bf L}_{V_k})\ge\left(\frac1{k!}\right)^2\cdot\mathrm{det}({\bf L}_{W_k}).
\end{displaymath}
When $k\le\mathrm{rank}({\bf L})$, $W_k$ satisfies $\mathrm{det}({\bf L}_{W_k})>0$, and therefore
\begin{displaymath}
  \mathrm{det}({\bf L}_{V_k})>0, \quad k\le\mathrm{rank}({\bf L}).
\end{displaymath}
According to Equ. \eqref{dj:intep}, we have
\begin{displaymath}
  d_{j_k}^2=\frac{\mathrm{det}({\bf L}_{V_k})}{\mathrm{det}({\bf L}_{V_{k-1}})}.
\end{displaymath}
As a result, for $k=1,\dots,\mathrm{rank}({\bf L})$, $d_{j_k}>0$.
When $k=\mathrm{rank}({\bf L})+1$, $V_k$ contains $\mathrm{rank}({\bf L})+1$ items, and ${\bf L}_{V_k}$ is singular.
Therefore, $d_{j_k}=0$ for $k=\mathrm{rank}({\bf L})+1$.
\end{proof}

According to Theorem~\ref{thm:dj}, when kernel $\bf L$ is a low rank matrix,
Algorithm~\ref{sup:fast} returns at most $\mathrm{rank}({\bf L})$ items.
For Algorithm~\ref{diverse:nearby} with a sliding window,
according to Subsection \ref{di:update:nearby},
$d_i$ is non-decreasing after the earliest item is removed.
This allows for returning more items, and alleviates the numerical instability issue.

\section{More Simulation Results}

We have also compared the metric popularity-weighted recall (PW Recall) \cite{steck2011item} of different algorithms.
Its definition is
\begin{displaymath}
  \mathrm{PW\ Recall}=\frac{\sum_{u\in U}\sum_{t\in T_u}w_{t}\mathbb{I}_{t\in R_u}}{\sum_{u\in U}\sum_{t\in T_u}w_{t}},
\end{displaymath}
where $T_u$ is the set of relevant items in the test set, $w_t$ is the weight of item $t$
with $w_t\propto C(t)^{-0.5}$ where $C(t)$ is the number of occurrences of $t$ in the training set,
and $\mathbb{I}_P$ is the indicator function.
PW Recall measures both relevance and diversity.

\begin{figure}[t]
	\begin{center}
		\includegraphics[width=1.7in]{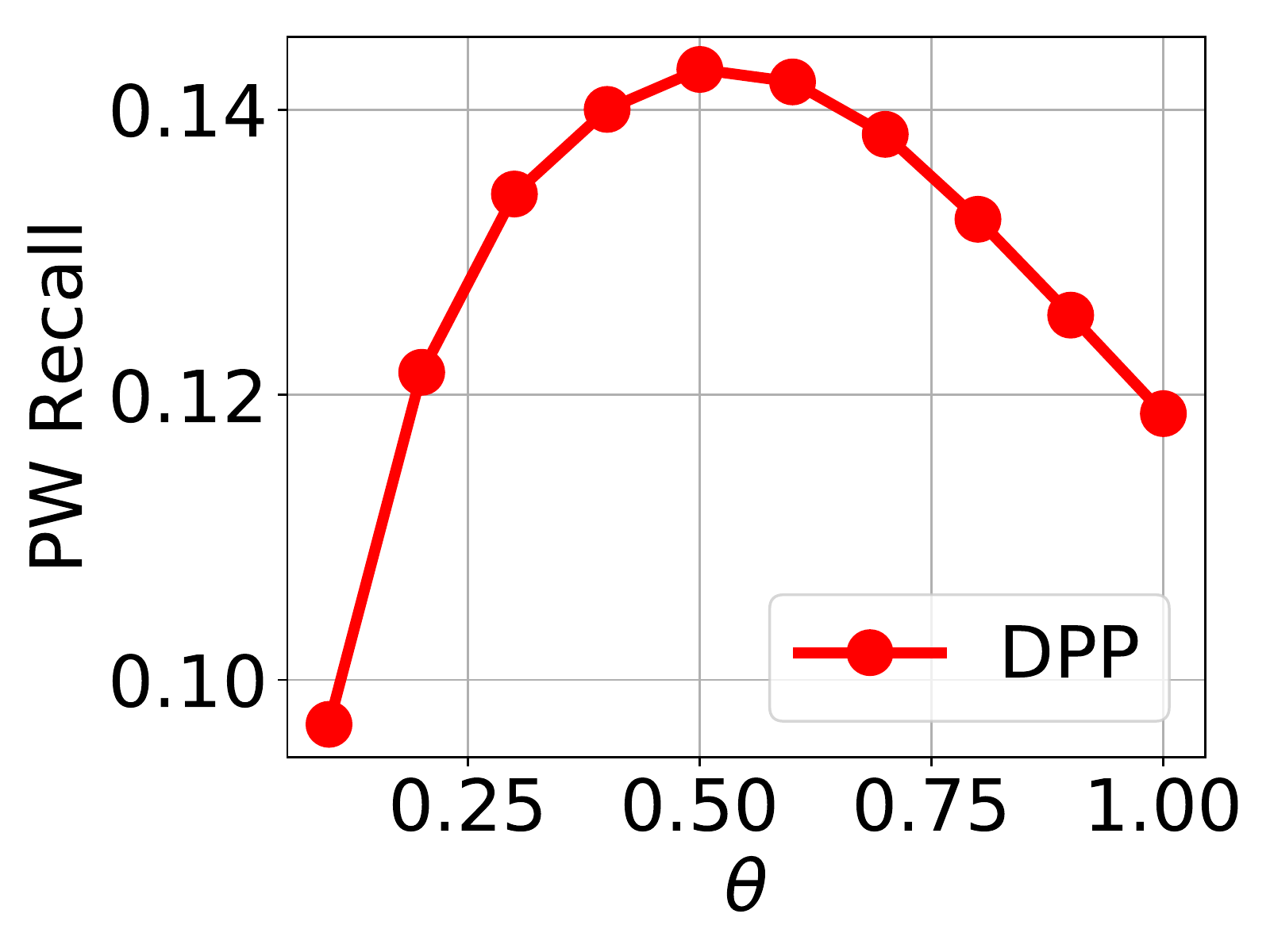}
		\caption{Impact of trade-off parameter $\theta$ on PW Recall on Netflix dataset.}
		\label{theta:compare:pw_recall}
	\end{center}
	\vskip -0.1in
\end{figure}

\begin{figure}[t]
	\begin{center}
		\begin{minipage}{.35\textwidth}
			\includegraphics[width=1.7in]{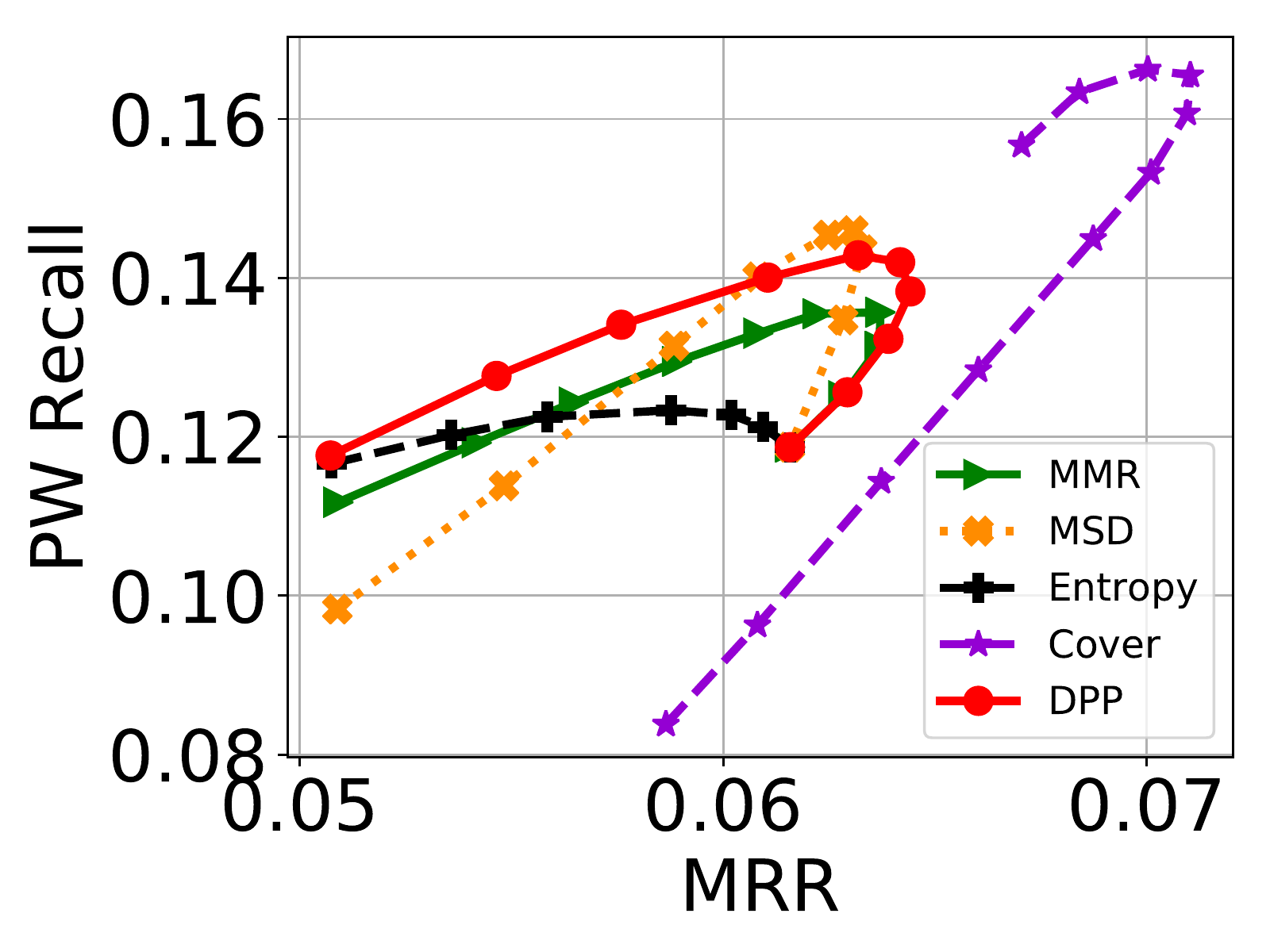}
		\end{minipage}
		\begin{minipage}{.35\textwidth}
			\includegraphics[width=1.7in]{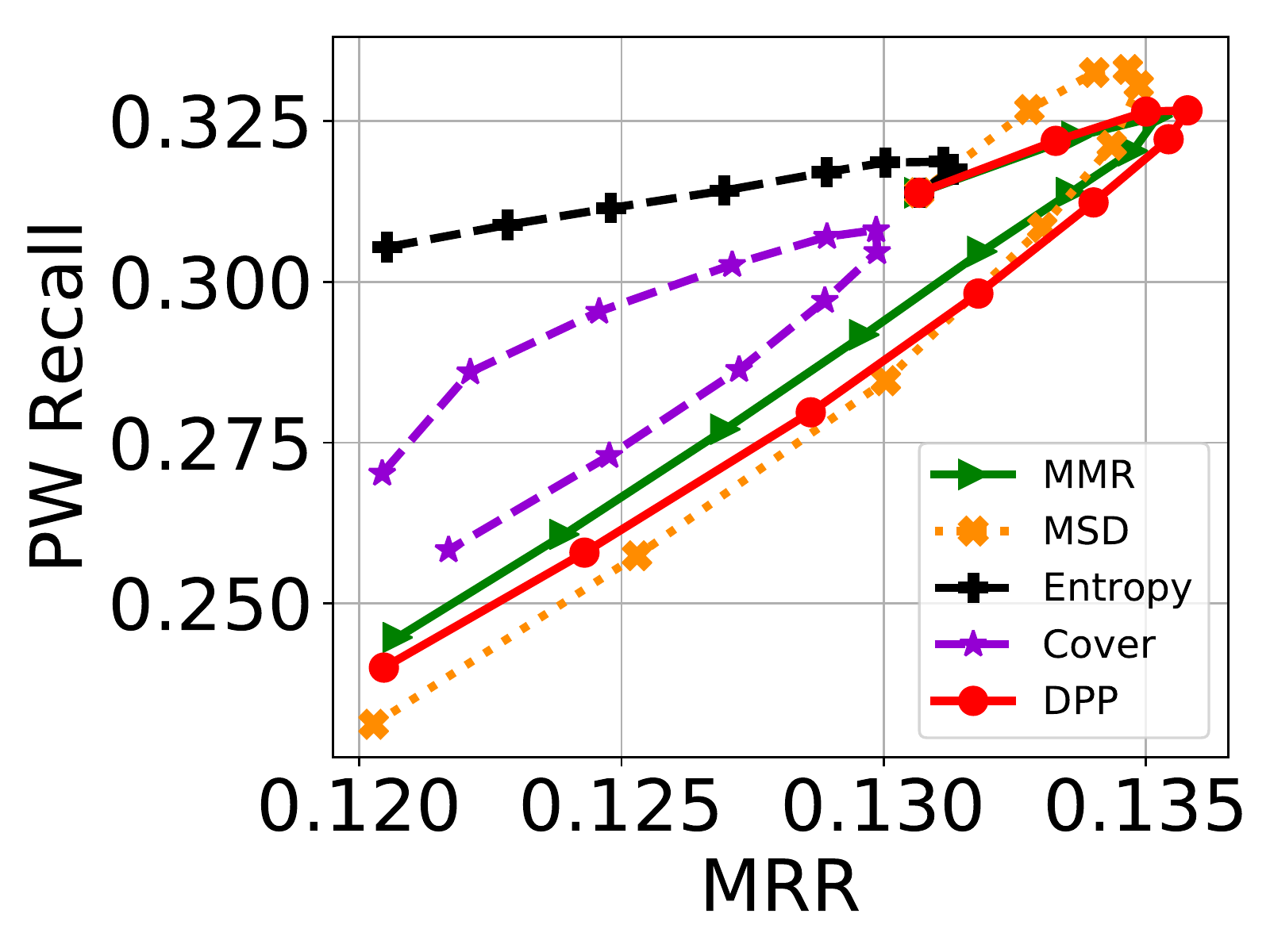}
		\end{minipage}
		\caption{Comparison of trade-off performance between MRR and PW Recall
			under different choices of trade-off parameters
			on Netflix Prize (left) and Million Song Dataset (right).}
		\label{one:out:pw_recall}
	\end{center}
	\vskip -0.05in
\end{figure}

\begin{figure}[!t]
	\begin{center}
		\begin{minipage}{.35\textwidth}
			\includegraphics[width=1.7in]{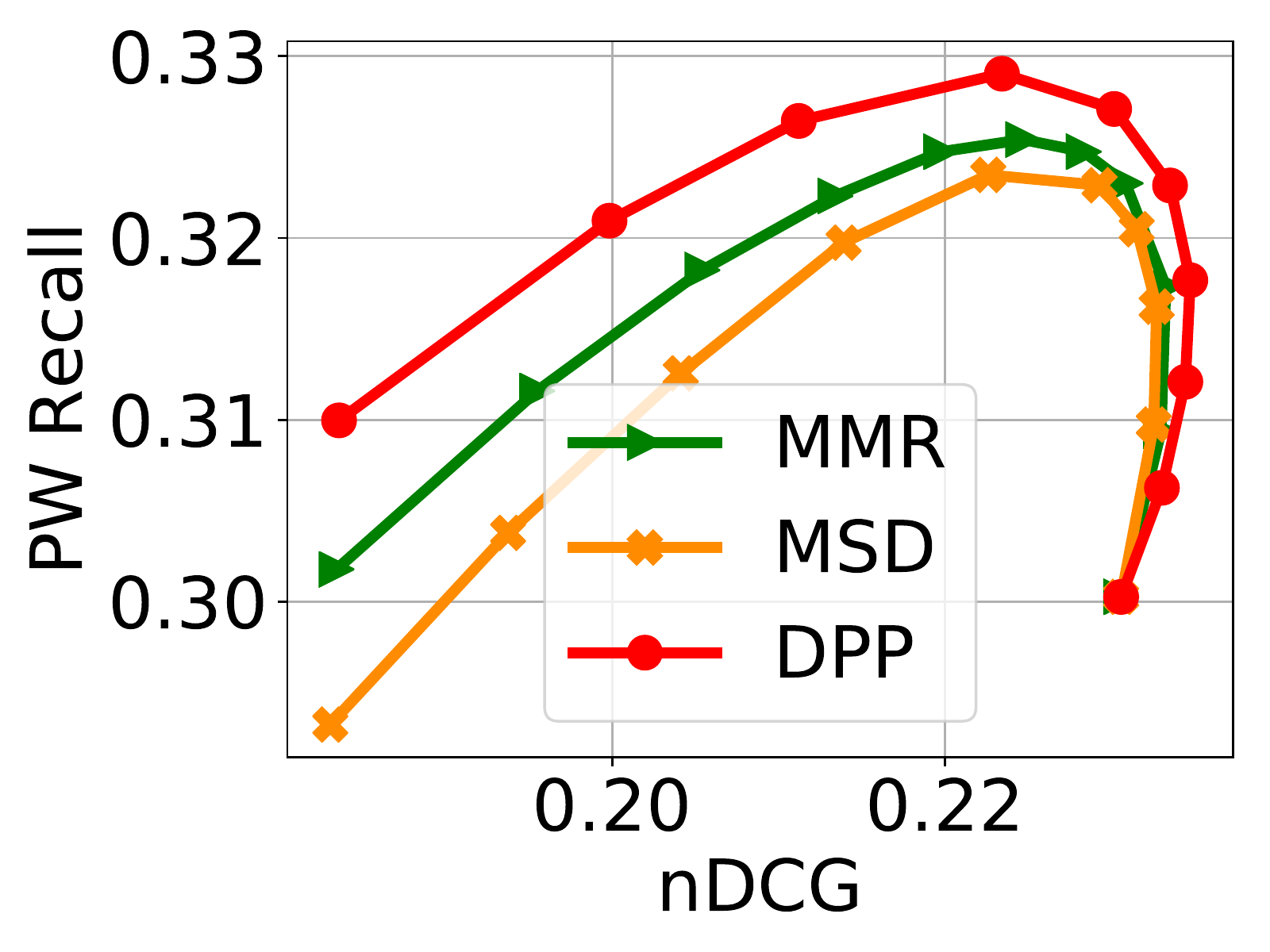}
		\end{minipage}
		\begin{minipage}{.35\textwidth}
			\includegraphics[width=1.7in]{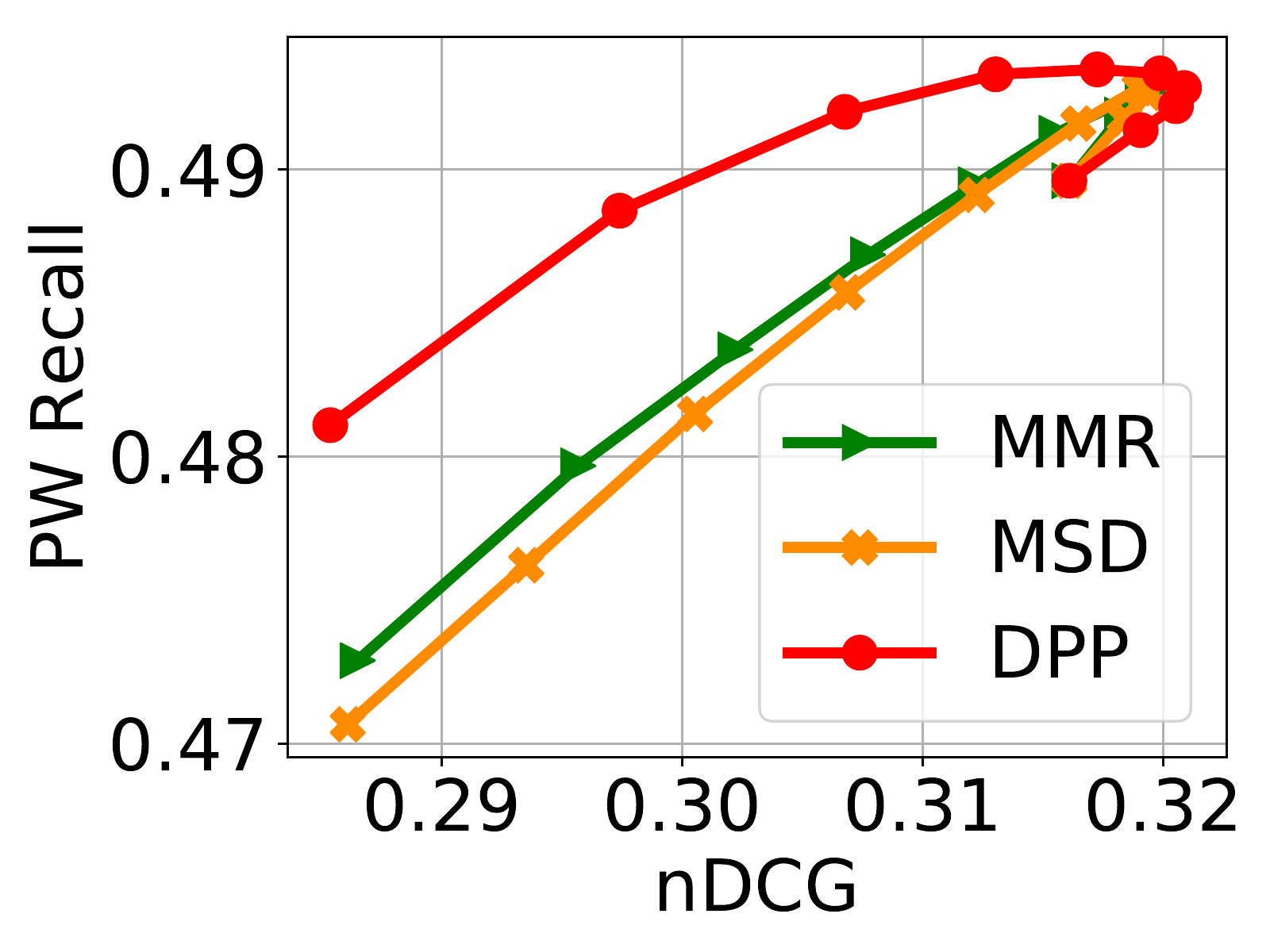}
		\end{minipage}
		\caption{Comparison of trade-off performance between nDCG and PW Recall
			under different choices of trade-off parameters on Netflix Prize (left) and Million Song Dataset (right).}
		\label{five:out:pw_recall}
	\end{center}
	\vskip -0.1in
\end{figure}

Similar to Figure~\ref{theta:compare}, the impact of trade-off parameter $\theta$ on PW Recall
on Netflix Prize is shown in Figure~\ref{theta:compare:pw_recall}.
As $\theta$ increases, MRR improves at first, achieves the best value when $\theta\approx0.5$, and then decreases.
Therefore, moderate amount of diversity also leads to better PW Recall.

Similar to Figure~\ref{one:out}, by varying the trade-off parameters,
the trade-off performance between MRR and PW Recall are compared in Figure~\ref{one:out:pw_recall}.
Similar conclusions can be drawn.

Similar to Figure~\ref{five:out}, by varying the trade-off parameters,
the trade-off performance between nDCG and PW Recall are compared in Figure~\ref{five:out:pw_recall}.
DPP enjoys the best trade-off performance on both datasets.

\end{document}